\UseRawInputEncoding
\documentclass[12pt, draftclsnofoot, journal, letterpaper, onecolumn]{IEEEtran}

\makeatletter
\def\ps@headings{%
\def\@oddhead{\mbox{}\scriptsize\rightmark \hfil \thepage}%
\def\@evenhead{\scriptsize\thepage \hfil \leftmark\mbox{}}%
\def\@oddfoot{}%
\def\@evenfoot{}}
\makeatother 

\IEEEoverridecommandlockouts
\usepackage{bbm}
\usepackage{amsfonts}
\usepackage[dvips]{graphicx}
\usepackage{times}
\usepackage{cite}
\usepackage{amsmath}
\usepackage{array}
\usepackage{amssymb}

\usepackage{stfloats}
\usepackage{slashbox}
\usepackage{graphicx}
\usepackage{subfigure}
\usepackage{footnote}
\usepackage{booktabs}
\usepackage{array}
\usepackage{algorithm}
\usepackage{subeqnarray}
\usepackage{cases}
\usepackage{threeparttable}
\usepackage{color}
\usepackage{epstopdf}
\usepackage{algpseudocode}
\usepackage{bm}
\usepackage{multirow}
\usepackage{adjustbox}

\newtheorem{theorem}{Theorem}
\newtheorem{remark}{Remark}

\newtheorem{lemma}{Lemma}

\newtheorem{corollary}{Corollary}
\newtheorem{assumption}{Assumption}
\allowdisplaybreaks[4]

\begin{document}

\title{Fundamental Limits of Communication Efficiency for Model Aggregation in Distributed Learning: A Rate-Distortion Approach}

\author{Naifu Zhang, Meixia Tao, Jia Wang and Fan Xu
\thanks{N. Zhang, M. Tao, J. Wang and F. Xu are with the Department of Electronic Engineering, Shanghai Jiao Tong University, Shanghai,
200240, P. R. China. (email: arthaslery@sjtu.edu.cn; mxtao@sjtu.edu.cn; jiawang@sjtu.edu.cn; xxiaof@sjtu.edu.cn.). Part of this work is presented at IEEE ISIT 2021 \cite{zhang2021sumratedistortion}.}}

\maketitle
\vspace{-1.5cm}
\begin{abstract}
One of the main focuses in distributed learning is communication efficiency, since model aggregation at each round of training can consist of millions to billions of parameters.
Several model compression methods, such as gradient quantization and sparsification, have been proposed to improve the communication efficiency of model aggregation.
However, the information-theoretic minimum communication cost for a given distortion of gradient estimators is still unknown.
In this paper, we study the fundamental limit of communication cost of model aggregation in distributed learning from a rate-distortion perspective.
By formulating the model aggregation as a vector Gaussian CEO problem, we derive the rate region bound and sum-rate-distortion function for the model aggregation problem, which reveals the minimum communication rate at a particular gradient distortion upper bound.
We also analyze the communication cost at each iteration and total communication cost based on the sum-rate-distortion function with the gradient statistics of real-world datasets.
It is found that the communication gain by exploiting the correlation between worker nodes is significant for SignSGD, and a high distortion of gradient estimator can achieve low total communication cost in gradient compression.
\end{abstract}


\section{Introduction}
In a wide range of artificial intelligence (AI) applications such as image recognition and natural language processing, the size of training datasets has grown significantly over the years due to the growing computation and sensing capabilities of mobile devices.
It is becoming crucial to train big machine learning models in a distributed fashion in which large-scale datasets are distributed over multiple worker machines for parallel processing.
Distributed learning offers several distinct advantages compared with traditional learning at a centralized data center, such as preserving privacy, reducing network congestion, and leveraging distributed on-device computation.
In a typical distributed learning framework such as federated learning \cite{mcmahan2016communication,konen2016federated,bonawitz2019towards,yang2019federated}, each worker node downloads a global model from the parameter server and computes an update by learning from its local dataset using, for instance, the stochastic gradient descent (SGD) algorithm.
These local updates are then sent to the parameter server and to improve the global model.
The process is repeated until the algorithm converges.

The main bottleneck in distributed learning is the communication cost for model aggregation \cite{dean2012large,seide20141,strom2015scalable} since the model update from each worker node at each round of training may consist of millions to billions of parameters.
This is exacerbated in the cases of the federated learning paradigm \cite{mcmahan2016communication,konen2016federated} and the cloud-edge AI systems \cite{8736011}, where the wireless channel bandwidth between edge server and edge devices is limited.
The total communication cost in distributed learning depends on the communication cost at each iteration and the number of iterations.
To reduce the communication cost at each iteration, one conventional method is to let each worker node compress its local update, e.g., stochastic gradient, via quantization and sparsification before aggregation.
Partially initiated by the 1-bit implementation of SGD by Microsoft in \cite{seide20141}, a large number of recent studies have revisited the idea of gradient quantization \cite{pmlr-v80-bernstein18a,alistarh2018qsgd,NIPS2017_89fcd07f}.
Other approaches for low-precision training focus on the sparsification of gradients, either by thresholding small entries or by random sampling \cite{aji2017sparse,lin2018deep}.
There also exist several approaches that combine quantization and sparsification to maximize performance gains, including Quantized SGD (QSGD)\cite{alistarh2018qsgd} and TernGrad\cite{NIPS2017_89fcd07f}.
By doing so, the distortion of the gradient received at the parameter server increases, which results in a lower convergence rate \cite{bubeck2015convex,wang2018cooperative,zhang2015deep}.
Hence, the training phase requires more number of iterations to maintain the target model accuracy.
These methods \cite{seide20141,pmlr-v80-bernstein18a,alistarh2018qsgd,NIPS2017_89fcd07f,aji2017sparse,lin2018deep} aim to reconstruct the model update computed from each of worker nodes.
However, by recalling that the main objective in distributed learning is a good estimate of the global model update at the parameter server by using the local model updates from worker nodes.
This implies that the exact recovery of each individual local model update in existing works \cite{seide20141,pmlr-v80-bernstein18a,alistarh2018qsgd,NIPS2017_89fcd07f,aji2017sparse,lin2018deep} is not necessary.
Specifically, in distributed learning, the objective is to estimate the global model update computed by gradient descent (GD) on a global dataset, while the local model update computed by each worker node is a noisy version of the global model update.
Moreover, the local model updates are highly correlated among different worker nodes, providing an opportunity for distributed source coding to reduce the communication cost in model aggregation.
However, this is not leveraged by these works \cite{seide20141,pmlr-v80-bernstein18a,alistarh2018qsgd,NIPS2017_89fcd07f,aji2017sparse,lin2018deep}.

As such, a fundamental question arises: what is the information-theoretic minimum communication cost required to achieve a target convergence bound (i.e., model accuracy) in distributed learning?
Considering that the total communication cost in distributed learning is determined by the number of learning iterations and the communication cost at each iteration, this question can be decoupled into two sub-questions:
1) For a given per-iteration distortion bound on gradient estimation, what is the required number of total iterations to achieve a target convergence bound?
2) what is the minimum communication rate at each iteration to achieve a target gradient distortion bound?
Then, by combining the solutions of the sub-questions and searching all possible distortion bounds on the gradient, we can find the global minimum communication cost to achieve a target convergence bound.
The former question has been solved by \cite{bubeck2015convex}, which provides closed-form convergence bounds of SGD algorithms.
The latter question essentially falls into the classic rate-distortion problem for lossy source coding.

The model aggregation in distributed learning with SGD algorithms is inherently a multiterminal remote source coding problem, or the so-called chief executive officer (CEO) problem \cite{berger1996ceo}, and the justification will be presented in Section III in detail.
In the CEO problem, the main objective of the CEO is to make a good estimate of a data sequence $\{X(t)\}_{t=1}^\infty$ that cannot be observed directly.
The CEO employs a team of $K$ agents who can observe the independently corrupted version of the data sequence, denoted as $\{Y_k(t)\}_{t=1}^\infty$, where $k=1,2,...,K$.
The observations $\{Y_k(t)\}_{t=1}^\infty$, for $k=1,2,...,K$, are separately encoded and forwarded to the CEO over rate-constrained channels.
The CEO then decodes the data sequence $\{X(t)\}_{t=1}^\infty$ from the received information from $K$ agents.
The quadratic Gaussian CEO problem, where the data sequence is assumed to be Gaussian distributed, first studied by Viswanathan and Berger \cite{viswanathan1997quadratic}, has received particular attention.
Oohama in \cite{669162} made major progress on the quadratic Gaussian CEO problem and showed that the sum-rate of the Berger-Tung achievable region \cite{10016434852} is tight.
This problem is further studied in \cite{oohama2005rate,1365154,5508637}, where the entire rate-distortion region is established in \cite{1365154}.
The achievability is obtained by using the Berger-Tung inner bound \cite{10016434852} while the converse, first developed in \cite{669162} and later refined in \cite{oohama2005rate,1365154}, is obtained by the application of Shannon's entropy power inequality to related various information quantities.

However, note that none of the existing rate region results for the CEO problem can be directly applied to distributed learning for the following reason.
The distortion measure in the original CEO problem is only mean square error (MSE) without any consideration on estimation bias.
On the other hand, the convergence rate of distributed learning depends on both the bias and variance of the gradient estimator \cite{ajalloeian2020convergence}.
This implies that MSE alone is not enough to measure the distortion of gradient estimation in distributed learning for convergence analysis.
As a result, the existing rate-distortion results for the CEO problem without bias constraint are not directly applicable.
For an example, in the classic rate-distortion function \cite{669162}, if the MSE $D$ is larger than the variance $\sigma_X^2$ of the Gaussian data sequence $\{X(t)\}^\infty_{t=1}$, the minimum achievable rate is zero, meaning that the agents do not transmit anything.
This is valid in the CEO problem from the pure estimation perspective, but obviously trivial in distributed learning as the model will not converge.
Based on the discussion above, this paper imposes the unbiased estimation constraint on the gradient estimator and re-visit the rate-distortion function.
By having the unbiased estimation constraint, the MSE-based distortion can then fully determine the convergence rate.
In addition, given a variance bound, the unbiased estimator achieves the maximum convergence rate\cite{ajalloeian2020convergence}.

Motivated by the above issue, in this paper, we study the fundamental limit of communication cost of model aggregation in distributed learning to achieve a target convergence bound from a rate-distortion perspective.
The primary goal of this paper is to provide an information-theoretic characterization of the communication cost at a given bounded variance and to have an insight into communication efficiency for distributed learning.
To the best of our knowledge, this is the first work aimed at analysing the rate-distortion function for model aggregation in distributed learning.
The main contributions of this work are outlined below:
\begin{itemize}
\item\emph{Problem formulation and justification:}
We formulate the model aggregation as a vector Gaussian CEO problem based on the gradient distribution.
Specifically, we model the global and local gradients as independent Gaussian vector sequences with zero mean.
We impose an unbiased constraint on the gradient estimator such that the MSE distortion can fully determine the convergence rate.
We also justify the assumptions of gradient distribution through experiments on real-world datasets.

\item\emph{Rate-distortion analysis:}
Based on the above problem formulation, we derive the optimal rate region at given distortion for the model aggregation problem.
For the achievability, we first design an unbiased estimator at the receiver and then adopt the Berger-Tung scheme \cite{10016434852}.
For the converse, we tighten the traditional converse bound \cite{1365154} by linear minimum mean squared error and Jensen's inequality.
It is proved that the achievable bound coincides with the converse bound, thus achieves optimality.
We provide an explicit formula for the rate region boundary and a closed-form sum-rate-distortion function by solving convex optimization problems.
Moreover, the formulated optimization problems can also obtain the optimal rate allocation in practical applications by simply replacing the coefficients.

\item\emph{Communication cost analysis:}
We analyze the communication cost at each iteration based on the gradient statistics of real-world datasets.
Numerical results demonstrate that the communication gain from exploiting the correlation between the local gradients of worker nodes is significant for SignSGD \cite{pmlr-v80-bernstein18a}.
We also analyze the total communication cost required to achieve a target convergence bound and guide the distortion selection.
It is found that a moderately high distortion of the gradient estimator achieves the minimum total communication cost.
\end{itemize}

The rest of this paper is organized as follows.
In Section II, we introduce the basics of distributed learning and its convergence rate.
In Section III, we formulate the model aggregation from the rate-distortion perspective.
In Section IV, we provide the optimal rate region result and its explicit formula.
Section V analyzes the communication cost based on real-world datasets.
We conclude the paper in Section VII.

\section{Basics of Distributed Learning}
In this section, we first introduce the basic distributed learning setup in error-free communication and then present its convergence bound as a function of the number of iterations and the variance bound of the gradient estimator.

\subsection{System Model}
\begin{figure}[t]
\begin{centering}
\vspace{-0.2cm}
\includegraphics[scale=.40]{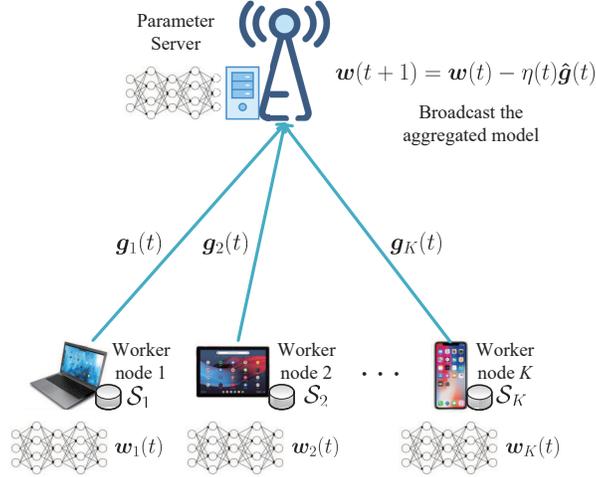}
\vspace{-0.1cm}
 \caption{\small{Illustration of a basic distributed learning system.}}\label{fig:system}
\end{centering}
\vspace{-0.3cm}
\end{figure}
We consider a basic distributed learning framework as illustrated in Fig.~\ref{fig:system}, where a global AI model (e.g., a neural network) is trained collaboratively across $K$ worker nodes via the coordination of a parameter server.
Let $\mathcal{K}=\{1,...,K\}$ denote the set of worker nodes.
Each worker node $k\in\mathcal{K}$ collects a fraction of labelled training data via interaction with its own users, constituting a local dataset, denoted as $\mathcal{S}_k$.
Let $\bm{w}\in\mathbb{R}^P$ denote the $P$-dimensional model parameter to be learned.
The loss function measuring the model error is defined as
\begin{equation}
F(\bm{w})=\sum_{k\in\mathcal{K}}{\frac{|\mathcal{S}_k|}{|\mathcal{S}|}F_k(\bm{w})},
\label{equ:loss_function}
\end{equation}
where $F_k(\bm{w})=\frac{1}{|\mathcal{S}_k|}\sum_{i\in\mathcal{S}_k}{f_i(\bm{w})}$ is the loss function of worker node $k$ quantifying the prediction error of the model $\bm{w}$ on the local dataset $\mathcal{S}_k$, with $f_i(\bm{w})$ being the sample-wise loss function, and $\mathcal{S}=\bigcup_{k\in\mathcal{K}}{\mathcal{S}_k}$ is the union of all datasets.
The minimization of $F(\bm{w})$ is typically carried out through the mini-batch SGD algorithm, where worker node $k$'s local dataset $\mathcal{S}_k$ is split into mini-batches of size $B_k$ and at each iteration $t=1,2,...,T$, we draw one mini-batch $\mathcal{B}_k(t)$ randomly and calculate the local gradient vector as
\begin{equation}
\bm{g}_k(t)=\nabla\frac{1}{B_k}\sum_{i\in\mathcal{B}_k(t)}{f_i(\bm{w}(t))}.
\label{equ:mini_batch_SGD}
\end{equation}
In the special case if the mini-batch size is $B_k = 1$, the algorithm reduces to the SGD algorithm.
Let $\bm{g}(t)\triangleq\nabla F(\bm{w}(t))\in\mathbb{R}^P$ denote the global gradient vector as if GD algorithm can be applied over the global dataset $\mathcal{S}$.
Given the global dataset $\mathcal{S}$ and the model parameter $\bm{w}$, $\bm{g}(t)$ is deterministic at each iteration $t$ but unknown to the parameter server.
By the definition of the mini-batch SGD algorithm, $\bm{g}_k(t)$ is a random vector satisfying $\mathbb{E}[\bm{g}_k(t)]=\bm{g}(t)$, where the expectation is over mini-batches.
Let $\bm{\tilde{g}}(t)$ denote the stochastic gradient obtained by SGD algorithm with $B_k=1$ and assume that its variance with respect to the global gradient $\bm{g}(t)$ at iteration $t$ is $\mathbb{E}[\|\bm{\tilde{g}}(t)-\bm{g}(t)\|^2]=\sigma^2(t)$.
By definition (\ref{equ:mini_batch_SGD}), the local gradient $\bm{g}_k(t)$ has variance $\frac{\sigma^2(t)}{B_k}$.
If all the local gradients $\{\bm{g}_k(t)\}_{k=1}^K$ are received by the parameter server through error-free transmission, the optimal estimator of $\bm{g}(t)$ to minimize its variance should be the Sample Mean Estimator, i.e., $\bm{\hat{g}}(t)=\sum_{k=1}^K{\frac{B_k}{B}\bm{g}_k(t)}$, where $B\triangleq\sum_{k=1}^K{B_k}$ is the global batch size.
Note that the estimator $\bm{\hat{g}}(t)$ is equivalent to a mini-batched stochastic gradient of batch size $B$, hence it is an unbiased estimation of $\bm{g}(t)$ with mean $\mathbb{E}[\bm{\hat{g}}(t)]=\bm{g}(t)$ and variance $\frac{\sigma^2(t)}{B}$.
Then the parameter server updates the model parameter as
\begin{equation}
\bm{w}(t+1)=\bm{w}(t)-\eta(t)\bm{\hat{g}}(t),
\label{equ:centralized_update}
\end{equation}
with $\eta(t)$ being the learning rate at iteration $t$.

\subsection{Convergence Rate}
Given access to the unbiased gradient estimator $\bm{\hat{g}}(t)$ with variance bound $D$, and the number of iterations $T$, the following theorem provides a convergence bound of distributed learning:
\begin{theorem}[\cite{bubeck2015convex}, Theorem 6.3]
\label{theorem:convergence_rate}
Let $\omega\subseteq\mathbb{R}^P$ be a convex set, and let $F:\omega\rightarrow\mathbb{R}$ be convex and $\beta$-smooth.
Let $\bm{w}(0)$ be given, and let $A^2=\sup_{\bm{w}\in\omega}{\|\bm{w}-\bm{w}(0)\|^2}<\infty$.
Let $T>0$ be fixed.
Given repeated, independent access to the unbiased gradient estimators with variance bound $D$ for loss function $F$, i.e., $\mathbb{E}[\|\bm{\hat{g}}(t)-\bm{g}(t)\|^2]\leq D$ for all $t=1,2,...,T$, initial point $\bm{w}(0)$ and constant step sizes $\eta(t)=\frac{\gamma}{\beta+1}$, where $\gamma=A\sqrt{\frac{2}{DT}}$, the convergence bound of the distributed learning satisfies

\begin{align}
\mathbb{E}\left[F\left(\frac{1}{T}\sum_{t=1}^T{\bm{w}(t)}\right)\right]-\min_{\bm{w}\in\mathbb{R}^P}{F(\bm{w})}\leq A\sqrt{\frac{2D}{T}}+\frac{\beta A^2}{T}.\label{equ:convergence_rate}
\end{align}
\end{theorem}
\begin{remark}
By Theorem \ref{theorem:convergence_rate}, we can directly obtain the required number of iterations $T$ to achieve the target convergence bound for a given variance bound $D$ of the gradient estimator.
Therefore, in order to find out the fundamental limits of communication efficiency for model aggregation in distributed learning, the remaining problem is to characterize the minimum communication cost at each iteration for a given $D$, which is a basic problem in rate-distortion theory and will be solved in the following sections.
Note that Theorem \ref{theorem:convergence_rate} is for the unbiased gradient estimator.
The convergence results of biased SGD methods (i.e., gradient sparsification) are analyzed in existing works, e.g., \cite{ajalloeian2020convergence}, but the model generally converges to a neighborhood of the optimal solution.
In this work, we focus on the unbiased gradient estimator.
\end{remark}

\section{Problem Formulation}
In this section, we formulate the model aggregation from a rate-distortion perspective and justify the assumptions of gradient distribution used in the formulated problem through experiments on real-world datasets.

\subsection{Rate-Distortion Problem for Model Aggregation}
We consider the model aggregation for a general $L$-layer deep neural network (DNN) model at each iteration $t$.
To facilitate the following derivation, we omit the iteration index and assume that the dimension size of each layer of the model is identical.
Note that all the local gradients $\{\bm{g}_k\}$ and the global gradient $\bm{g}$ in distributed learning are unknown to the parameter server, thus each of them is viewed as a realization of a random source with certain distribution.
As stated in \cite{NIPS2017_89fcd07f}, the value range of gradients in different layers can be different as gradients are back propagated.
This makes it inaccurate to treat the entire gradient vector $\bm{g}_k(t)$ in one iteration as one sequence of i.i.d random variables with length $P$.
To this end, we segment each gradient vector at each iteration into a sequence of random vectors, where the dimension of each vector corresponds to the number of layers of the DNN, $L$, and the length of the sequence is determined by the number of neurons in each DNN layer, given by $P/L$.
The $P/L$ vectors in this sequence can be assumed independent and identically distributed and jointly encoded in practice.

Based on the gradient distribution justification in the next subsection, we define the global gradient as an independent Gaussian vector sequence $\{\bm{X}(i)\}_{i=1}^n$ with mean $0$ and covariance matrix $\mbox{diag}(\sigma_{X_1}^2,\sigma_{X_2}^2,...,\sigma_{X_L}^2)$, where $\sigma_{X_l}^2$ is the variance of the $l$-th layer of the global gradient, $n$ is the sequence length and $L$ is the number of layers.
Each $\bm{X}(i),i=1,2,...,n$ takes value in real space $\mathbb{R}^L$.
Note that $\bm{g}\in\mathbb{R}^P$ in Section II can be viewed as a realization of $\{\bm{X}(i)\}_{i=1}^n$ with $P=nL$.
Similarly, we define each local gradient as an independent Gaussian vector sequence $\{\bm{Y}_k(i)\}_{i=1}^n$, for $k=1,2,...,K$, which is a noisy version of $\{\bm{X}(i)\}_{i=1}^n$, each taking value in real space $\mathbb{R}^L$ and corrupted by independent additive white Gaussian noise, i.e.,
\begin{align}
\bm{Y}_k(i)=\bm{X}(i)+\bm{N}_k(i),i=1,2,...,n,
\end{align}
where $\bm{N}_k(i)$ is a Gaussian random vector with mean 0 and covariance matrix $\mbox{diag}(\sigma_{N_{k,1}}^2,...,\sigma_{N_{k,L}}^2)$.
We assume that $\bm{N}_k$ is independent over worker node $k$.
Note that, for $k=1,2,...,K$, $\bm{g}_k\in\mathbb{R}^P$ in Section II can be viewed as a realization of $\{\bm{Y}_k(i)\}_{i=1}^n$
In the following, we write $n$ independent copies of $\bm{X}$ and $\bm{Y}_k$, i.e., $\{\bm{X}(i)\}_{i=1}^n$ and $\{\bm{Y}_k(i)\}_{i=1}^n$, as $\bm{X}^n$ and $\bm{Y}_k^n$, respectively.

\begin{figure}[t]
\begin{centering}
\vspace{-0.2cm}
\includegraphics[scale=.40]{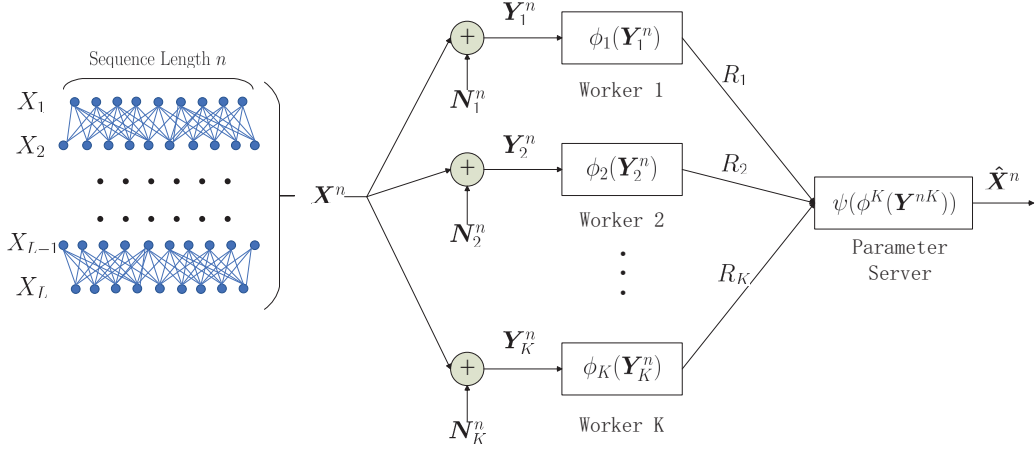}
\vspace{-0.1cm}
 \caption{\small{Model aggregation formulation from rate-distortion theory perspective.}}\label{fig:CEO}
\end{centering}
\vspace{-0.3cm}
\end{figure}
The proposed problem formulation for model aggregation is shown in Fig.~\ref{fig:CEO}.
The parameter server is interested in the sequence $\bm{X}^n$ (global gradient) that cannot be observed directly.
The parameter server employs $K$ worker nodes, and each worker node $k$, for $k=1,2,...,K$, observes independently corrupted version $\bm{Y}_k^n$ (local gradient) of $\bm{X}^n$.
Each local gradient $\bm{Y}_k^n$ is separately encoded to $\phi_k(\bm{Y}_k^n)$ and sent to the parameter server.
Upon receiving $\phi_k(\bm{Y}_k^n)$, for $k=1,2,...,K$, the parameter server outputs the estimation of $\bm{X}^n$ by using the decoder function $\hat{\bm{X}}^n=\psi(\phi_1(\bm{Y}_1^n),\phi_2(\bm{Y}_2^n),...,\phi_K(\bm{Y}_K^n))$.
The encoder function $\phi_k$ of each worker node $k$, for $k=1,2,...,K$, is defined by
\begin{align}
\phi_k:\mathbb{R}^{L\times n}\rightarrow\mathcal{C}_k\triangleq\{1,2,...,2^{nR_k}\},
\label{equ:rate_constrain}
\end{align}
where $R_k$ is the communication rate of worker node $k$.
The decoder function $\psi$ of the parameter server is defined by
\begin{align}
\psi:\mathcal{C}_1\times\mathcal{C}_2\times...\times\mathcal{C}_K\rightarrow\mathbb{R}^{L\times n}.
\end{align}
We define a $K$-tuple encoder function of all the $K$ encoder functions as
\begin{align}
\phi^K=(\phi_1,\phi_2,...,\phi_K).
\end{align}
Similarly, we define
\begin{align}
\phi^K(\bm{Y}^{nK})=(\phi_1(\bm{Y}_1^n),\phi_2(\bm{Y}_2^n),...,\phi_K(\bm{Y}_K^n)),
\end{align}
where $\bm{Y}^{nK}$ is the collection of the local gradients computed by all the worker nodes.
To apply Theorem \ref{theorem:convergence_rate} for convergence analysis, we impose the following unbiased gradient estimator:
\begin{align}
\mathbb{E}[\hat{\bm{X}}^n|\bm{X}^n=\bm{x}^n]=\bm{x}^n,
\label{equ:unbiased_definition}
\end{align}
and for $\hat{\bm{X}}^n=\psi(\phi^K(\bm{Y}^{nK}))$, define the MSE distortion by
\begin{align}
D(\bm{X}^n,\hat{\bm{X}}^n)=\frac{1}{n}\sum_{i=1}^n{\mathbb{E}[\|\bm{X}(i)-\hat{\bm{X}}(i)\|^2]}.
\label{equ:distortion}
\end{align}
The distortion (\ref{equ:distortion}) is also the variance of the unbiased gradient estimator $\hat{\bm{X}}^n$ and it can fully determine the convergence rate by Theorem \ref{theorem:convergence_rate}.
Note that in traditional rate-distortion problems, the only target is to minimize the distortion measured by MSE (\ref{equ:distortion}) without the constraint (\ref{equ:unbiased_definition}).
Given the MSE distortion, the gradient estimator has uncertain bias and variance, hence the MSE distortion alone cannot fully determine the convergence rate.

For a given distortion $D$, a rate $K$-tuple $(R_1,R_2,...,R_K)$ is said to be achievable if there are encoders $\phi^K$ satisfying (\ref{equ:rate_constrain}) and a decoder $\psi$ such that $\bm{\hat{X}}^n$ is unbiased (\ref{equ:unbiased_definition}) and $D(\bm{X}^n,\hat{\bm{X}}^n)\leq D$ for some $n$.
The closure of the set of all achievable rate $K$-tuples is called the rate region and we denote it by $\mathcal{R}_\star(D)\subseteq\mathbb{R}_+^K$.
To find the information-theoretic minimum communication cost at each iteration for a given distortion bound is to characterize the region $\mathcal{R}_\star(D)$.
Note that the model aggregation problem is similar to the vector Gaussian CEO problem.
However, the rate region results of the vector Gaussian CEO problems cannot be applied directly to model aggregation because they do not consider the unbiased constraint (\ref{equ:unbiased_definition}).

\subsection{Gradient Distribution}
In this subsection, we use real-world datasets to justify the Gaussian assumptions of gradient distribution, which are used in the previous sub-section and are listed as follows:
\begin{assumption}[Gaussian distribution]
\label{assumption:distribution}
The global gradient $\bm{X}$ is a Gaussian vector with zero mean. The local gradients $\bm{Y}_k$ for $k=1,2,...,K$ are noisy versions of the global gradient $\bm{X}$, i.e., $\bm{Y}_k=\bm{X}+\bm{N}_k$ and the gradient noises $\bm{N}_k$ for $k=1,2,...,K$ are also Gaussian vectors with zero mean.
\end{assumption}
\begin{assumption}[Noise independence]
\label{assumption:noise_independence}
Each gradient noise $\bm{N}_k$ is independent of the global gradient $\bm{X}$ and gradient noises $\bm{N}_k$ for $k=1,2,...,K$ are independent of each other.
\end{assumption}
We train multilayer perceptron (MLP) model on the MNIST dataset to justify the two assumptions above.
The hyper-parameters are set as follows: momentum optimizer is 0.5, local batch size is 128, and learning rate is 0.01.

\subsubsection{Justification of Assumption \ref{assumption:distribution}}
We show the empirical probability density function (PDF) of gradients to justify the Gaussian Distribution.
\begin{figure}[t]
\centering

\subfigure[The PDF of global gradient over $6$ dimensions.]
{\begin{minipage}[t]{0.48\linewidth}
\centering
\includegraphics[width=3.2in]{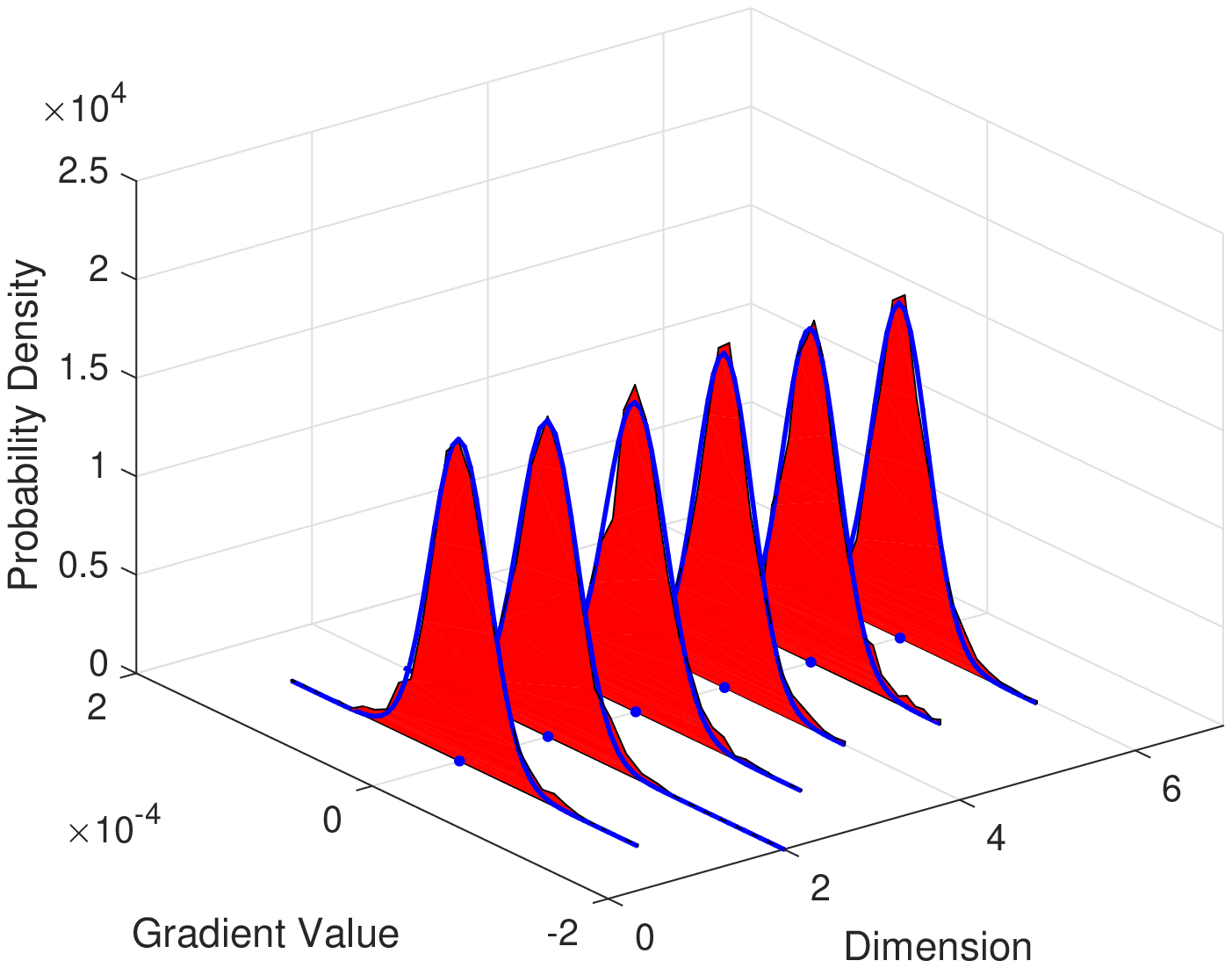}
\label{fig:gradient_distribution_over_dimension}
\end{minipage}}
\subfigure[The PDF of global gradient over iteration $1-10$.]
{\begin{minipage}[t]{0.48\linewidth}
\centering
\includegraphics[width=3.2in]{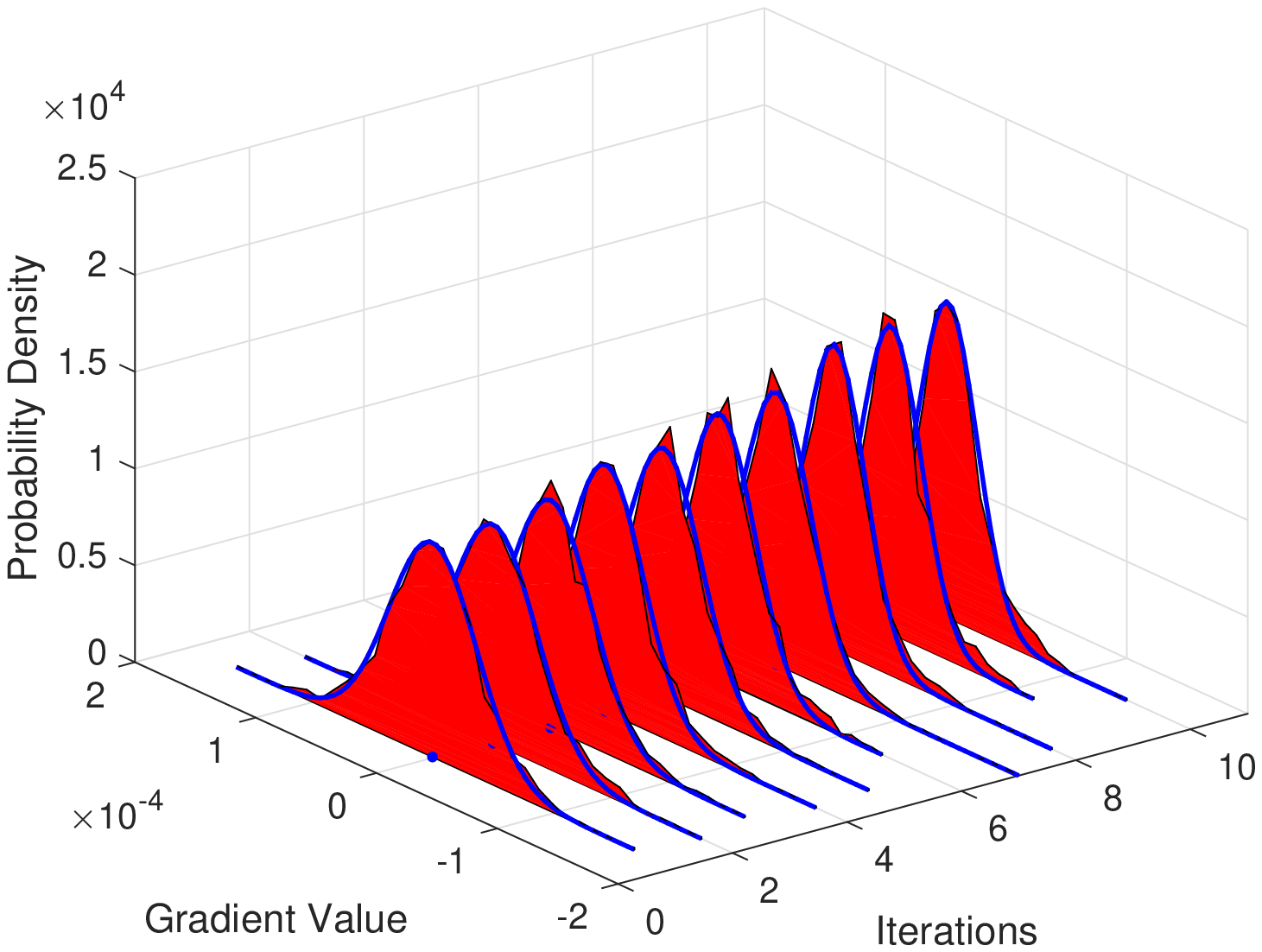}
\label{fig:gradient_distribution_over_iteration}
\end{minipage}}

\centering
\caption{The empirical PDF of global gradient over 2500 samples and its Gaussian fitting curve, where the mean of each Gaussian fitting curve is marked with blue dots.}
\label{fig:gradient_distribution} 
\end{figure}
Fig.~\ref{fig:gradient_distribution} illustrates the empirical PDF of global gradient over 2500 samples and its Gaussian fitting curve, where the mean of each Gaussian fitting curve is marked with blue dots and the square of correlation coefficient of the fitting is 0.99.
Recall that the global gradients are calculated through the GD algorithm.
We randomly initialize the model $\bm{w}(0)$ and update the model by GD on global dataset for $10$ iterations, and this process is repeated 2500 times.
It is observed from Fig.~\ref{fig:gradient_distribution} that the global gradient $\bm{X}$ follows a Gaussian distribution with zero mean.
Gaussian distribution of gradients is consistent with the observation in existing works \cite{NIPS2017_89fcd07f,sra2012optimization}.
In Fig.~\ref{fig:gradient_distribution_over_iteration}, it is observed that the variance of global gradients gradually decreases with iteration.
Intuitively, the magnitude of the global gradient realization decreases when the model approaches convergence.
The smaller the magnitude of a realization, the smaller the variance of the corresponding random variable.
Hence, the variance of global gradients decreases with iteration.
Note that the model converges only if the global gradient has zero mean and zero variance.

\begin{figure}[t]
\begin{centering}
\vspace{-0.2cm}
\includegraphics[scale=.50]{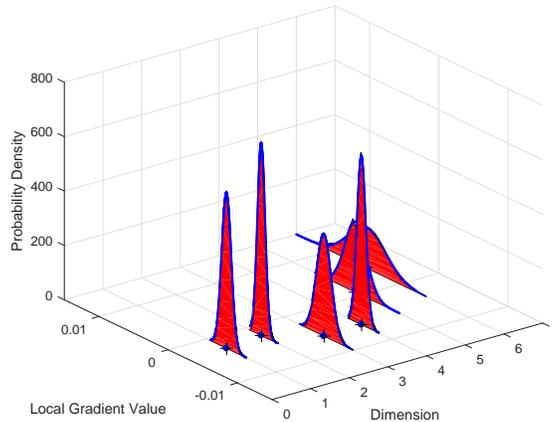}
\vspace{-0.1cm}
 \caption{\small{The empirically conditional PDF of local gradient $\bm{Y}_k$ given global gradient $\bm{X}=\bm{x}$ over 10000 samples and its Gaussian fitting curve, where the mean of each Gaussian fitting curve and the given global gradient realization $\bm{x}$ are marked with blue dot and black "+", respectively.}}
 \label{fig:local_gradient_distribution}
\end{centering}
\vspace{-0.3cm}
\end{figure}
Fig.~\ref{fig:local_gradient_distribution} illustrates empirically conditional PDF of local gradient $\bm{Y}_k$ given global gradient $\bm{X}=\bm{x}$ over 10000 samples and its Gaussian fitting curve, where the mean of each Gaussian fitting curve is marked with blue dot, the given global gradient realization $\bm{x}$ is marked with black "+" and the square of correlation coefficient of the fitting is 0.99.
Recall that each local gradient $\bm{Y}_k$ is calculated through the mini-batch SGD algorithm at worker node $k$.
The given global gradient realization $\bm{x}$ is obtained by GD based on the global dataset with a fixed model $\bm{w}$ at iteration 10.
The empirical distribution of $\bm{Y}_k$ at the given realization $\bm{x}$ is obtained by randomly drawing 10000 local batches.
It is observed that the conditional local gradient $\bm{Y}_k$ given $\bm{X}=\bm{x}$ follows a Gaussian distribution with mean $\bm{x}$, which justifies that the local gradient $\bm{Y}_k$ is a noisy version of global gradient $\bm{X}$ and corrupted by Gaussian noise with zero mean.

\subsubsection{Justification of Assumption \ref{assumption:noise_independence}}
We adopt Pearson correlation coefficients to verify the independence of gradient noises.
\begin{table}[t]
\centering
\caption{\small{Pearson correlation coefficient between gradient noise and global gradient, where the number of samples is 1000.}}
\label{table:noise_independent_over_global_gradient}
\begin{tabular}{|c|c|c|c|c|c|c|c|c|c|c|}
\hline
Iteration      & 1      & 2      & 3      & 4      & 5      & 6      & 7      & 8      & 9      & 10      \\ \hline
Dimension 1 & 0.0075 & -0.0064 & 0.0009 & 0.0123 & -0.023 & 0.0092 & -0.004 & -0.0069 & -0.0029 & 0.0046 \\ \hline
Dimension 2 & -0.0064 & -0.0110 & -0.0051 & -0.0152 & 0.0060 & 0.0046 & 0.0109 & -0.0110 & -0.0076 & -0.0014 \\ \hline
Dimension 3 & 0.0123 & -0.0051 & 0.0045 & 0.0115 & -0.0016 & 0.0166 & -0.0007 & 0.0301 & -0.0021 & -0.0028 \\ \hline
Dimension 4 & -0.0239 & -0.0152 & 0.0143 & 0.0045 & 0.0127 & 0.0166 & 0.0063 & -0.0072 & -0.0068 & -0.0073 \\ \hline
\end{tabular}
\end{table}
Pearson correlation coefficient between gradient noise and global gradient is shown in Table~\ref{table:noise_independent_over_global_gradient}, where the number of samples is 1000.
It is observed that the absolute value of Pearson correlation coefficient between gradient noise and global gradient is less than 0.03.
In general, two variables can be considered weakly correlated or uncorrelated when the absolute value of correlation coefficient between them is less than 0.1.
Hence, we can assume that each gradient noise $\bm{N}_k$ is independent of global gradient $\bm{X}$.
This conforms with the intuition that the gradient noise depends on the selected local batch and the selection is independent of global gradient.

\begin{figure}[t]
\begin{centering}
\vspace{-0.2cm}
\includegraphics[scale=.50]{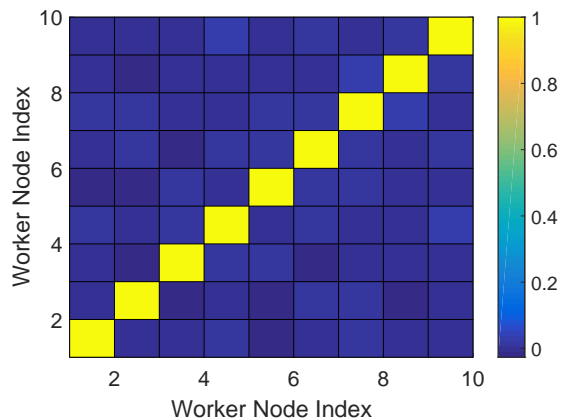}
\vspace{-0.1cm}
 \caption{\small{Absolute value of Pearson correlation coefficient of gradient noises between worker nodes, where the number of worker nodes is $K=10$ and the number of samples is 1000.}}
 \label{fig:noise_independent_over_worker_node}
\end{centering}
\vspace{-0.3cm}
\end{figure}
Fig.~\ref{fig:noise_independent_over_worker_node} illustrates absolute value of the Pearson correlation coefficient of gradient noises between worker nodes, where the number of worker nodes is $K=10$.
For each worker node, we sample 1000 local gradients and calculate the corresponding gradient noises.
It is observed that the Pearson correlation coefficient of gradient noises between different worker nodes are nearly zero (all of them are less than 0.02).
Hence, we can assume that the gradient noises $\bm{N}_k$ for $k=1,2,...,K$ are independent of each other.
This also conforms with the intuition that the gradient noise depends on the selected local batch and each worker node selects local batch independently.

\section{Rate Region Results}
In this section, we derive the optimal rate region for the model aggregation problem.
We first consider the models with single layer, and then extend the results to the models with multiple layers.

\subsection{Single-Layer Model}
In this case, the gradient vectors $\bm{X}$ and $\bm{Y}_k$, for $k=1,2,...,K$, reduce to scalar $X$ and $Y_k$, for $k=1,2,...,K$, respectively.
The optimal rate region result is provided in the following theorem:

\begin{theorem}
\label{theorem:rate_region}
For a given distortion $D$, the optimal rate region $\mathcal{R}_\star(D)$ is given by
\begin{align}
\mathcal{R}_\star(D)=\bigcup_{(r_{1},r_{2},...,r_{K})\in\mathcal{F}(D)}{\mathcal{R}(r_{1},r_{2},...,r_{K})},
\end{align}
where
\begin{align}
&\mathcal{R}(r_{1},r_{2},...,r_{K})=\bigg\{(R_{1},R_{2},...,R_{K}):\nonumber\\
&\sum_{k\in\mathcal{A}}{R_{k}}\geq\sum_{k\in\mathcal{A}}{r_{k}}+\frac{1}{2}\log{(\frac{1}{\sigma_{X}^2}+\frac{1}{D})}-\frac{1}{2}\log{\left(\frac{1}{\sigma_{X}^2}+\sum_{k\in\mathcal{A}^c}{\frac{1-\exp(-2r_{k})}{\sigma_{N_{k}}^2}}\right)},\forall \mathcal{A}\subseteq\mathcal{K}\bigg\},
\label{region:set_R}
\end{align}
and
\begin{align}
\mathcal{F}(D)=\left\{(r_{1},r_{2},...,r_{K})\in\mathbb{R}_+^K:\sum_{k=1}^K{\frac{1-\exp(-2r_{k})}{\sigma_{N_{k}}^2}}=\frac{1}{D}\right\}.
\label{region:set_r}
\end{align}
\end{theorem}
\begin{proof}
We only present the main idea here, while the complete proof is given in Appendix \ref{proof:one_layer_achievability} and Appendix \ref{proof:one_layer_converse}.
For the achievability part, we adopt the Berger-Tung scheme but design an unbiased estimator at the receiver.
For the converse part, we need to tighten the traditional converse bound derived in \cite{1365154} to match the newly proposed achievable scheme.
The challenge lies in deriving a larger bound for the term $\frac{1}{n}I(X^n;\hat{X}^n)$.
The traditional bound $\frac{1}{2}\log(\frac{\sigma_X^2}{D})$ for the term $\frac{1}{n}I(X^n;\hat{X}^n)$ can be easily derived by only applying MSE of $\hat{X}^n$, but it is not tight enough under the unbiased constraint.
We tighten this bound to $\frac{1}{2}\log(1+\frac{\sigma_X^2}{D})$ by first deriving the conditional variance of $X|\hat{X}$ via linear minimum mean squared error and then obtaining $D(X^n,\hat{X}^n)$ via Jensen's inequality.
\end{proof}

\begin{remark}
Parameter $r_{k}$ can be interpreted as the rate of the worker node $k$ for quantizing its gradient noise.
Recall that, in Appendix \ref{proof:one_layer_converse}, we tighten the bound $\frac{1}{2}\log(\frac{\sigma_X^2}{D})$ in (\cite{1365154},~Eq.8) to $\frac{1}{2}\log(1+\frac{\sigma_X^2}{D})$, hence our rate region $\mathcal{R}_\star(D)$ with unbiased estimation constraint is a subset of the classic rate region result in \cite{1365154}.
For example, the classic rate region achieves zero when the distortion $D$ is larger than the global gradient variance $\sigma_{X}^2$.
This indicates that the gradient estimator can be simply set to $\hat{X}=0$ at the receiver without any transmission.
However, $\hat{X}=0$ is not an option in our unbiased setting since $\mathbb{E}[\hat{X}|X=x]=0$ does not satisfy (\ref{equ:unbiased_definition}).
In this case, our rate region under unbiased estimator is still bounded away from zero.
As a result, for a distortion $D$ larger than the global gradient variance $\sigma_{X}^2$, the model will never converge by applying the classic achievable scheme while can converge by applying the proposed achievable scheme under unbiased estimator.
\end{remark}

\subsection{Multiple-Layers Model}
The following corollary extends Theorem \ref{theorem:rate_region} to the case where the model has multiple layers:
\begin{corollary}
\label{corollary:rate_region}
\begin{align}
\mathcal{R}_\star(D)=\bigg\{(R_1,...,R_K)\in\mathbb{R}_+^K:&R_k\geq\sum_{l=1}^L{R_{k,l}},k=1,2,...,K\nonumber\\
&\forall (R_{1,l},R_{2,l},...,R_{K,l})\in\mathcal{R}_l(r_{1,l},r_{2,l},...,r_{K,l}),l=1,2,...,L\nonumber\\
&\forall (r_{1,l},r_{2,l},...,r_{K,l})\in\mathcal{F}_l(D_l),l=1,2,...,L\nonumber\\
&\forall (D_1,D_2,...,D_L)\in\mathbb{R}_+^L,\sum_{l=1}^L{D_l}=D\bigg\}\label{region:D_l}
\end{align}
where
\begin{align}
&\mathcal{R}_{l}(r_{1,l},r_{2,l},...,r_{K,l})=\bigg\{(R_{1,l},R_{2,l},...,R_{K,l}):\nonumber\\
&\sum_{k\in\mathcal{A}}{R_{k,l}}\geq\sum_{k\in\mathcal{A}}{r_{k,l}}+\frac{1}{2}\log{(\frac{1}{\sigma_{X_l}^2}+\frac{1}{D_l})}-\frac{1}{2}\log{\left(\frac{1}{\sigma_{X_l}^2}+\sum_{k\in\mathcal{A}^c}{\frac{1-\exp(-2r_{k,l})}{\sigma_{N_{k,l}}^2}}\right)},\forall \mathcal{A}\subseteq\mathcal{K}\bigg\},
\label{region:set_R_l}
\end{align}
and
\begin{align}
\mathcal{F}_l(D_l)=\left\{(r_{1,l},r_{2,l},...,r_{K,l})\in\mathbb{R}_+^K:\sum_{k=1}^K{\frac{1-\exp(-2r_{k,l})}{\sigma_{N_{k,l}}^2}}=\frac{1}{D_l}\right\}.
\end{align}
\end{corollary}
\begin{proof}
Please refer to Appendix \ref{proof:corollary:rate_region}.
\end{proof}

\begin{remark}
Parameter $D_l$ can be interpreted as the distortion of the gradient estimator on the $l$-th layer, $r_{k,l}$ can be interpreted as the rate of the worker node $k$ for quantizing its observation noise on the $l$-th layer, and $R_{k,l}$ can be interpreted as the rate contributed by worker node $k$ on the $l$-th layer.
\end{remark}

\section{Explicit Formula for Rate Region}
The expression of the rate region in section IV is implicit.
To facilitate the analysis of the minimum communication cost of model aggregation, we derive the rate region boundary and a closed-form sum-rate-distortion function in this section.

\subsection{Rate Region Boundary}
We can explicitly compute the boundary of the rate region $\mathcal{R}_\star(D)$ by solving the following optimization problem,
\begin{align}
\min_{(R_1,R_2,...,R_K)\in\mathcal{R}_\star(D)}\sum_{k\in\mathcal{K}}\alpha_{k}R_k=\min_{D_l}\min_{r_{k,l}}\min_{R_{k,l}}{\sum_{l=1}^L\sum_{k\in\mathcal{K}}{\alpha_k R_{k,l}}},
\label{optimization:boundary}
\end{align}
for all choices of auxiliary coefficients $(\alpha_1,\alpha_2,...,\alpha_K)\in\mathbb{R}_+^K$.
Without loss of generality, we assume that $(\alpha_1\geq\alpha_2\geq...\geq\alpha_K)$.

\begin{lemma}
\label{lemma:optimal_R_kl}
Given $D_l$, $r_{k,l}$ and $(\alpha_1\geq...\geq\alpha_K)$, the optimal choice of $(R_{1,l},R_{2,l},...,R_{K,l})\in\mathcal{R}_{l}(r_{1,l},r_{2,l},...,r_{K,l})$ for minimizing $\sum_{k\in\mathcal{K}}{\alpha_k R_{k,l}}$ is
\begin{align}
R_{k,l}=r_{k,l}+\frac{1}{2}\log{\left(\frac{1}{\sigma_{X_l}^2}+\sum_{i=k}^K{\frac{1-\exp(-2r_{i,l})}{\sigma_{N_{i,l}}^2}}\right)}-\frac{1}{2}\log{\left(\frac{1}{\sigma_{X_l}^2}+\sum_{i=k+1}^K{\frac{1-\exp(-2r_{i,l})}{\sigma_{N_{i,l}}^2}}\right)}.
\label{equ:optimal_R_kl}
\end{align}
\end{lemma}
\begin{proof}
Please refer to Appendix \ref{proof:lemma:optimal_R_kl}.
\end{proof}

Based on Lemma \ref{lemma:optimal_R_kl}, we rewrite optimization (\ref{optimization:boundary}) in terms of parameters $z_{k,l}=\frac{1-\exp(-2r_{k,l})}{\sigma_{N_{k,l}}^2}$ as
\begin{subequations}
\begin{align}
\mathcal{P}_1:\quad\min_{z_{k,l},D_l}\quad&\sum_{l=1}^L\sum_{k=1}^K{\alpha_{k}\left[-\frac{1}{2}\log\left(1-\sigma_{N_{k,l}}^2z_{k,l}\right)+\frac{1}{2}\log{\left(1+\frac{z_{k,l}}{\frac{1}{\sigma_{X_l}^2}+\sum_{i=k+1}^K{z_{i,l}}}\right)}\right]}\\
\mathnormal{s.t.}\quad &~~~~~~~~~~~~~~~~~~~~\sum_{k\in\mathcal{K}}z_{k,l}\geq\frac{1}{D_l},~~l=1,2,...,L\label{full_constraint:D_l}\\
&~~~~~~~~~~~~~~~~~~~~\sum_{l=1}^L{D_l}\leq D\label{full_constraint:D}.
\end{align}
\end{subequations}
It is obvious that the above optimization problem is convex and the inequality constraints (\ref{full_constraint:D_l}) and (\ref{full_constraint:D}) are active, so it can be solved using Lagrange minimization.
Problem $\mathcal{P}_1$ can not only calculate the rate region boundary, but also solve the rate allocation problem in wireless edge learning by simply replacing $(\alpha_1,\alpha_2,...,\alpha_K)$ with other realistic parameters.
Take the model aggregation problem in wireless edge learning with channel fading as an example.
For a target gradient distortion, the goal is to find the optimal rate allocation to minimize the bandwidth consumption.
In this case, $\alpha_k$, for $k=1,2,...,K$, in Problem $\mathcal{P}_1$ can be replaced with the bandwidth consumed per rate, which can be calculated by the known channel gain, transmit power, and channel noise of worker node $k$.
Then, the solution to the rewritten Problem $\mathcal{P}_1$ is the optimal rate allocation.

\subsection{Sum-rate-distortion function}
In general, when considering the communication efficiency of distributed learning, such as model compress and aggregation frequency control, we only care about the total communication cost, i.e., the sum-rate-distortion function, which is defined by
\begin{align}
R_{sum}(D)\triangleq\min_{(R_1,R_2,...,R_K)\in\mathcal{R}_\star(D)}{\sum_{k\in\mathcal{K}}{R_k}}.
\label{definition:sum_rate_distortion_function}
\end{align}
The sum-rate-distortion function $R_{sum}(D)$ indicates the minimum communication cost required to achieve a particular gradient distortion $D$.
To facilitate the analysis of the communication gain obtained from exploiting the correlation of the local gradients between worker nodes, we assume that $L=1$ and the worker nodes are homogeneous with $\sigma_{N_1}^2=...=\sigma_{N_K}^2=\sigma_{N}^2$.
Then, we can obtain a closed-form expression for the sum-rate-distortion function in the following corollary, the proof of which is given in Appendix \ref{proof:corollary:sum_rate_distortion_function}.
\begin{corollary}
For every $D>0$, the sum-rate-distortion function is given by
\begin{align}
R_{sum}(D)=\frac{K}{2}\log\left(1+\frac{\sigma_{N}^2}{KD-\sigma_{N}^2}\right)+\frac{1}{2}\log(1+\frac{\sigma_{X}^2}{D}).
\label{equ:sum_rate_distortion_function}
\end{align}
\label{corollary:sum_rate_distortion_function}
\end{corollary}
\begin{remark}
The sum-rate in (\ref{equ:sum_rate_distortion_function}) consists of two terms.
The first term depends on the variance of gradient noise, $\sigma_N^2$ as well as the number of worker nodes, $K$.
It can be interpreted as the rate for quantizing $K$ independent gradient noises.
The second term depends on the variance of the global gradient, $\sigma_X^2$, which is the classical channel capacity for a Gaussian channel with a transmit power $\sigma_{X}^2$ and a noise power $D$.
It thus can be interpreted as the rate for quantizing the global gradient.
If each worker node transmits its own local gradient without distributed source coding, the sum-rate-distortion function can easily reduce to $R_{in}(D)=\frac{K}{2}\log\left(1+\frac{\sigma_{N}^2}{KD-\sigma_{N}^2}\right)+\frac{K}{2}\log(1+\frac{\sigma_{X}^2}{KD})$ by independently transmitting each $U_k$ in the achievability of Theorem \ref{theorem:rate_region}.
Note that the rate difference between $R_{in}(D)$ and $R_{sum}(D)$, i.e., $\frac{K}{2}\log(1+\frac{\sigma_{X}^2}{KD})-\frac{1}{2}\log(1+\frac{\sigma_{X}^2}{D})$, represents the communication gain by exploiting the correlation between worker nodes, which can be achieved by practical distributed source coding \cite{4494710,zhang2021adaptive}.
Based on the gradient statistics $\sigma_X^2(t)$ and $\sigma_N^2(t)$ on real-world datasets, we will analyze the optimal communication cost of distributed learning in the next section.
\end{remark}

\section{Communication Cost Analysis in Distributed Learning}
In this section, we analyse the communication cost of distributed learning based on real-world datasets.
We first demonstrate the communication gain of model aggregation by exploiting the correlation between worker nodes based on the sum-rate-distortion function (\ref{equ:sum_rate_distortion_function}).
Then, we analyze the total communication cost to achieve a target convergence bound of convex mini-batch SGD.
We also verify the communication cost in the non-convex case by numerical results.

We conducted simulation in a simulated environment where the number of worker nodes involved in each training iteration is $K=10$ if not specified otherwise.
We evaluate the model training on the MNIST, SVHN and CIFAR-10 datasets.
We adopt the multilayer perceptron (MLP) model on the MNIST dataset, the Resnet20 model on the SVHN dataset, and the Resnet50 model on the CIFAR-10 dataset.
The hyper-parameters are set as follows: momentum optimizer is 0.5, local batch size is 128 and learning rate $\eta=0.01$.

\subsection{Optimal Communication Cost over Iterations}
We first study the gradient statistics on real-world datasets.
\begin{figure}[t]
\centering
\subfigure[MLP model on MNIST.]
{\begin{minipage}[t]{0.32\linewidth}
\centering
\includegraphics[width=2.1in]{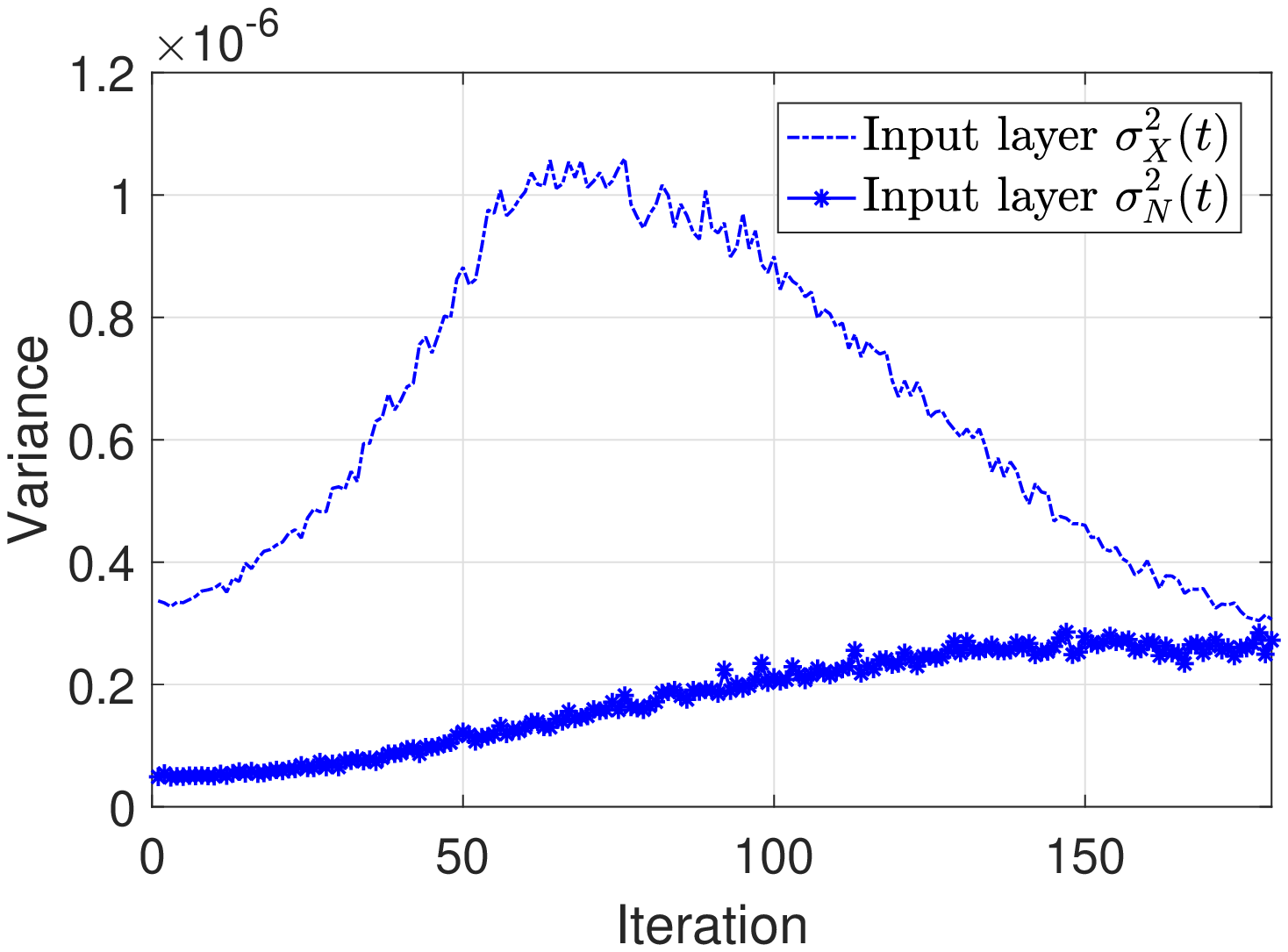}
\label{fig:statistic_over_iteration_mlp}
\end{minipage}}
\subfigure[Resnet20 model on SVHN.]
{\begin{minipage}[t]{0.32\linewidth}
\centering
\includegraphics[width=2.1in]{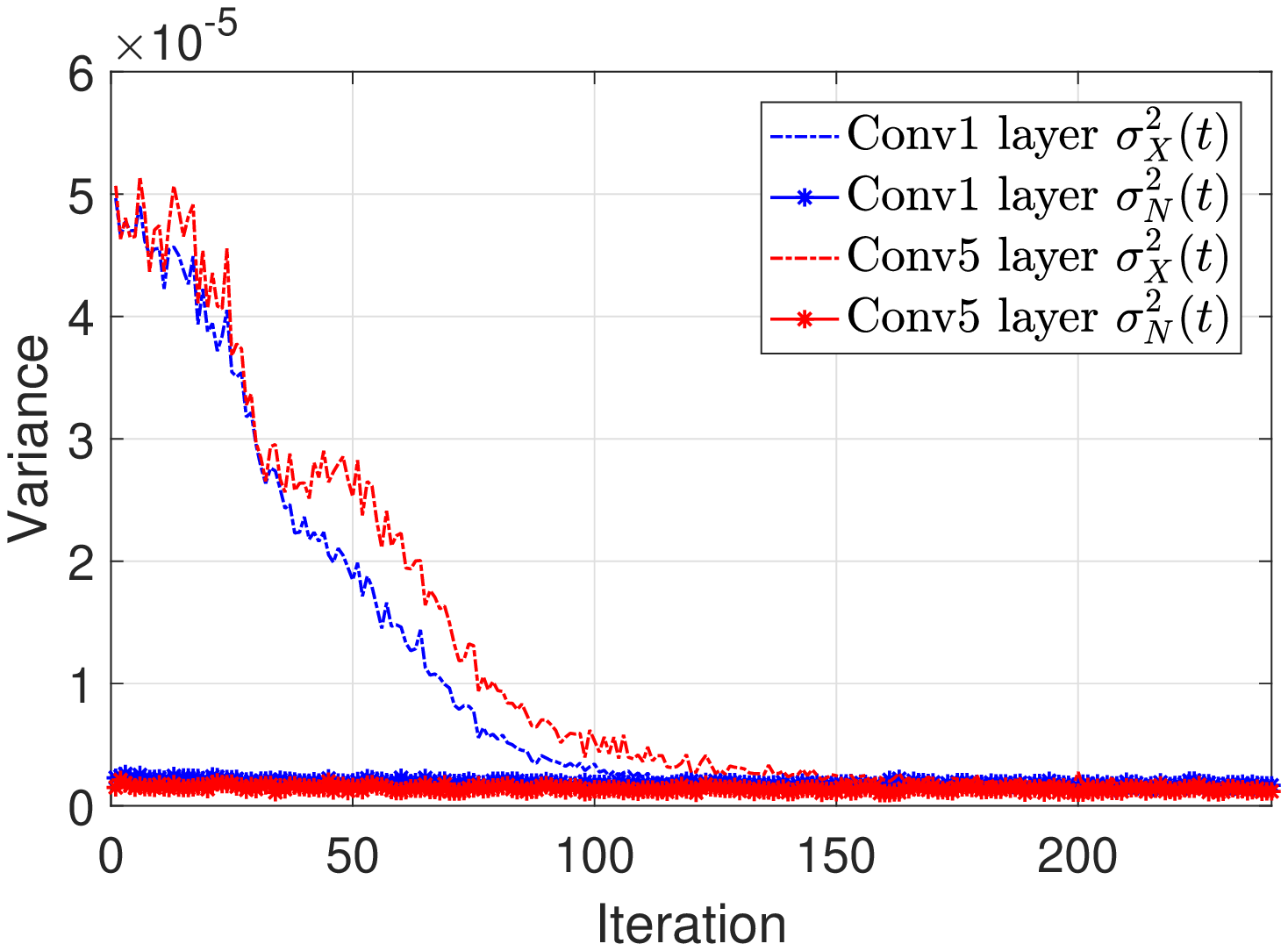}
\label{fig:statistic_over_iteration_resnet20}
\end{minipage}}
\subfigure[Resnet56 model on CIFAR-10.]
{\begin{minipage}[t]{0.32\linewidth}
\centering
\includegraphics[width=2.1in]{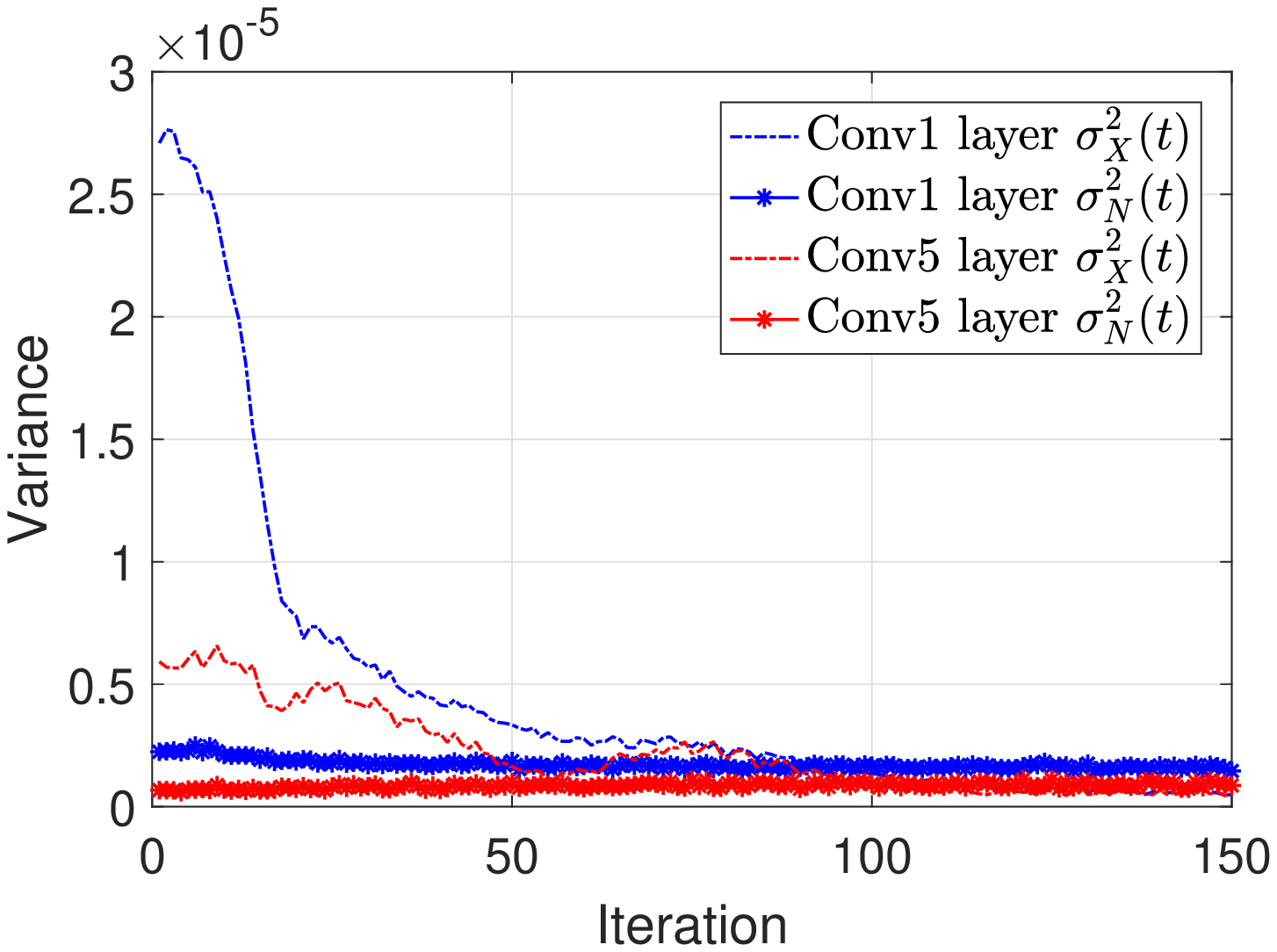}
\label{fig:statistic_over_iteration_resnet56}
\end{minipage}}
\centering
\caption{Gradient statistics $\sigma_X^2(t)$ and $\sigma_N^2(t)$ over iterations for three models, where the model are updated with ideal gradients without any transmission error.}
\label{fig:statistic_over_iteration} 
\end{figure}
Fig.~\ref{fig:statistic_over_iteration} shows the gradient statistics $\sigma_X^2(t)$ and $\sigma_N^2(t)$ over iterations for three models.
These three models approach convergence at the end of the last iteration.
The gradient statistics $\sigma_X^2(t)$ and $\sigma_N^2(t)$ are calculated offline based on the perfect local gradients.
It is observed that the global gradient variance $\sigma_X^2(t)$ is large at the beginning of the model training and gradually decreases over the iterations, while the gradient noise variance $\sigma_N^2(t)$ remains approximately unchanged for all three models.
This indicates that the correlation between the local gradients is strong at the beginning and gradually weakens over iterations.
Intuitively, in SGD-based learning, the initial model is far away from the convergent point at the beginning of the model training, and the common part in the local gradients, i.e., the global gradient, dominates, hence $\sigma_X^2(t)$ is large at the beginning.
When the model almost converges, the global gradient $\bm{g}(t)$ approaches zero, hence $\sigma_X^2(t)$ gradually decreases to zero.
On the other hand, the gradient noise variance $\sigma_N^2(t)$ remains approximately unchanged due to the random selection of the local batch throughout the training process.

Based on the above gradient statistics, the sum-rate-distortion function (\ref{equ:sum_rate_distortion_function}) can describe the optimal communication cost of model aggregation in distributed learning.
In the following, we analyze the communication redundancy of SignSGD \cite{pmlr-v80-bernstein18a} (a state-of-the-art gradient quantization scheme) by comparing the communication cost of SignSGD with $R_{in}(D)$ and $R_{sum}(D)$.
\begin{figure}[t]
\centering
\subfigure[MLP model on MNIST.]
{\begin{minipage}[t]{0.32\linewidth}
\centering
\includegraphics[width=2.1in]{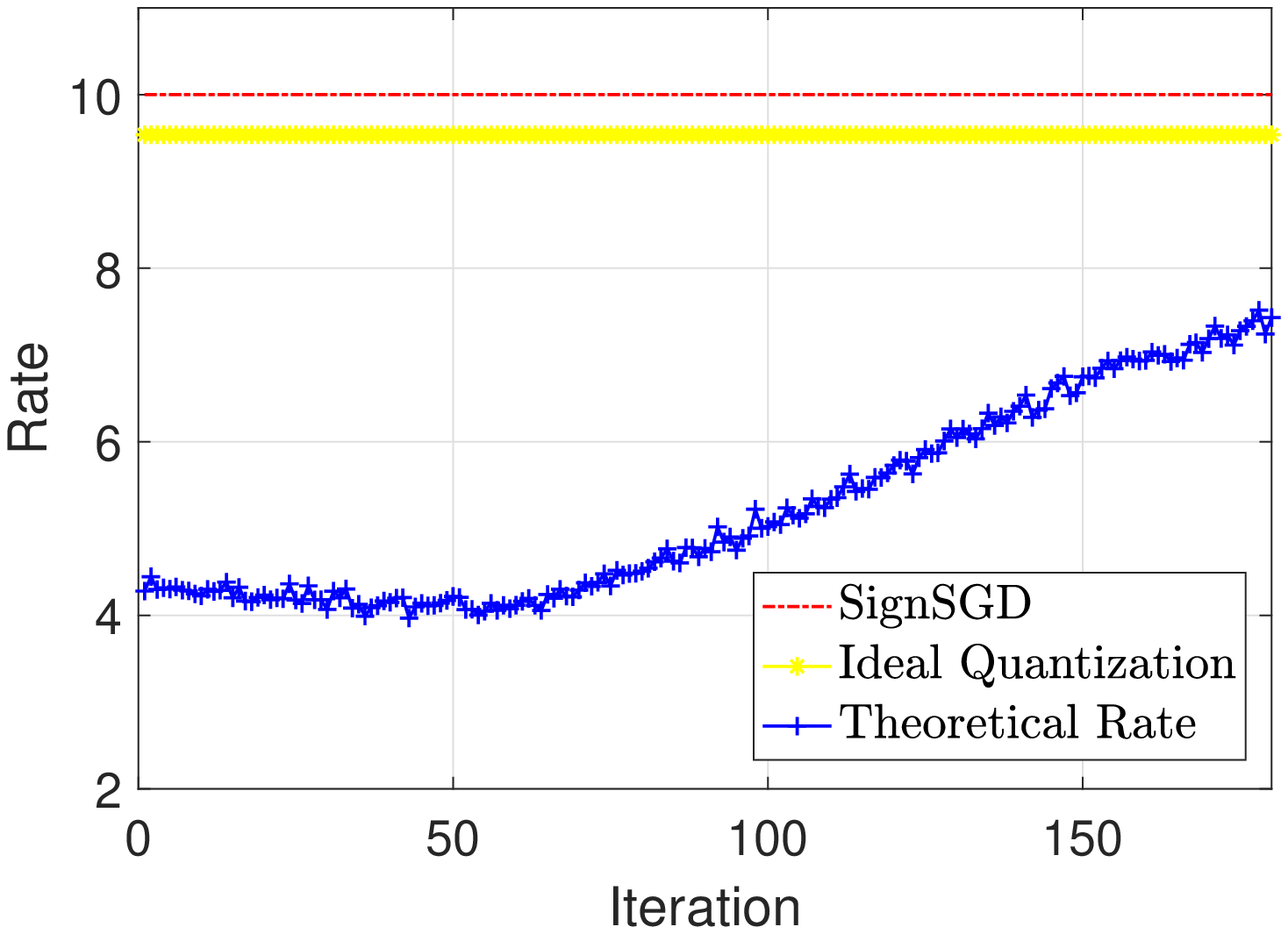}
\label{fig:rate_over_iteration_mlp}
\end{minipage}}
\subfigure[Resnet20 model on SVHN.]
{\begin{minipage}[t]{0.32\linewidth}
\centering
\includegraphics[width=2.1in]{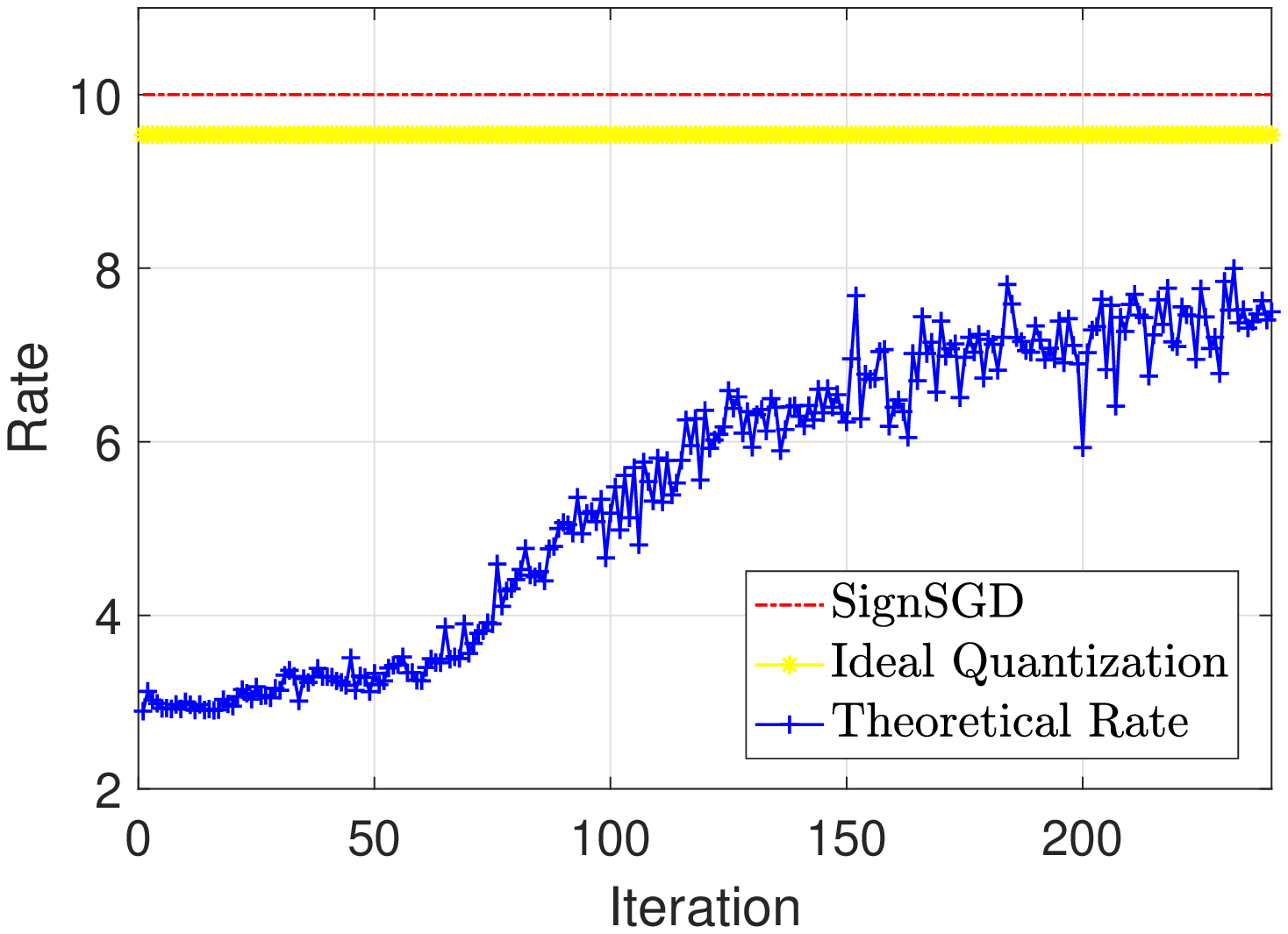}
\label{fig:rate_over_iteration_resnet20}
\end{minipage}}
\subfigure[Resnet56 model on CIFAR-10.]
{\begin{minipage}[t]{0.32\linewidth}
\centering
\includegraphics[width=2.1in]{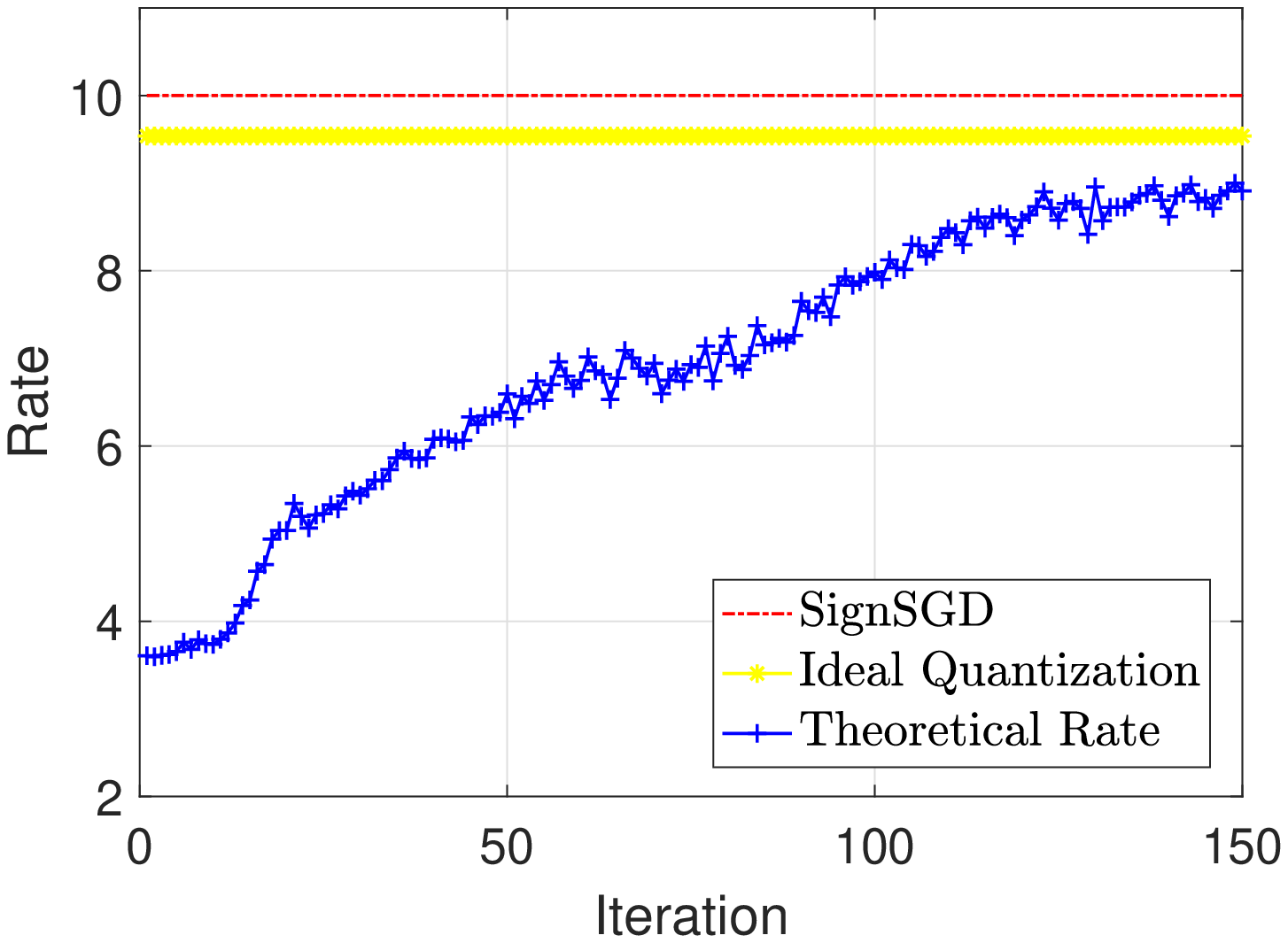}
\label{fig:rate_over_iteration_resnet56}
\end{minipage}}
\centering
\caption{The communication rate for SignSGD, the communication rate for ideal quantization, $R_{in}(D)$, and the theoretical minimum achievable rate, $R_{sum}(D)$, where the quantized value of SignSGD is $\{\pm\sqrt{\frac{2}{\pi}}\sigma\}$ and the distortion corresponding to the theoretical rate is equal to that of SignSGD.}
\label{fig:rate_over_iteration} 
\end{figure}
Fig.~\ref{fig:rate_over_iteration} shows the communication rate for SignSGD, the communication rate for ideal quantization, $R_{in}(D)$, and the theoretical minimum achievable rate, $R_{sum}(D)$, where the distortion $D$ remains the same for all the schemes.
SignSGD uses one bit to quantize each dimension of the local gradient, and we set the quantized value as $\{\pm\sqrt{\frac{2}{\pi}}\sigma\}$, which is optimal for the Gaussian sources with variance $\sigma^2$.
In distributed learning, the local gradients are Gaussian distributed with variance $\sigma_X^2+\sigma_N^2$, hence we have $\sigma=\sqrt{\sigma_X^2+\sigma_N^2}$.
The ideal quantization scheme can achieve the granular gain to minimize the quantization noise, but each worker node still transmits each quantized local gradient independently.
The distortion of SignSGD can be calculated by $D=\frac{1}{K}\left[\sigma_N^2+\frac{\pi-2}{\pi}(\sigma_N^2+\sigma_X^2)\right]$, then the corresponding sum rate $R_{in}(D)$ and $R_{sum}(D)$ can be calculated.
Note that substituting distortion $D$ of SignSGD into $R_{in}(D)$, we have $R_{in}=\frac{K}{2}\log(1+\frac{\pi}{\pi-2})$.
Hence, the communication rate of ideal quantization is a constant in Fig.~\ref{fig:rate_over_iteration}.
It is observed that there is a small gap between SignSGD and ideal quantization.
This is because scalar quantization encodes each dimension separately and does not exploit the granular gain.
This gap can be bridged by using vector quantization \cite{4494710}.
It is also observed that there is a significant gap between the ideal quantization scheme and the theoretical communication cost.
Moreover, the gap is large at the beginning of the training process.
This indicates that there is still redundant information in the quantized local gradients, especially at the beginning of model training, which is consistent with the observation in Fig.~\ref{fig:statistic_over_iteration}.
This gap can be bridged by a practical Slepian-Wolf coding for the quantized local gradients \cite{4494710,zhang2021adaptive}.

\subsection{Total Communication Cost}
Combining the sum-rate-distortion function in Corollary \ref{corollary:sum_rate_distortion_function} and the convergence bound in Theorem \ref{theorem:convergence_rate} for mini-batch SGD algorithms, we can directly yield the following result:
\begin{corollary}
\label{corollary:convex_communication_efficiency}
Let $\omega\subseteq\mathbb{R}^P$ be a convex set, and let $F:\omega\rightarrow\mathbb{R}$ be a convex and $\beta$-smooth function.
Let $\bm{w}(0)\in\omega$ be the initial point, and let $A^2=\sup_{\bm{w}\in\omega}{\|\bm{w}-\bm{w}(0)\|^2}<\infty$.
Let $\sigma_X^2(t)$ be the variance of the global gradient and $\sigma_N^2(t)$ be the variance of the gradient noise at iteration $t$.
The parameter server estimates the global gradient from $K$ identical worker nodes.
Suppose the model training accesses the gradient estimators $\{\bm{\hat{X}}(t)\}_{t=1}^T$ with distortion bound $D$, i.e., $\mathbb{E}\|\bm{X}(t)-\bm{\hat{X}}(t)\|^2\leq D$ for all $t=1,2,...,T$, and with step size $\eta(t)=\frac{\gamma}{\beta+1}$, where $\gamma$ is as in Theorem \ref{theorem:convergence_rate}.

For any given convergence bound $\epsilon>0$, in order to guarantee that $\mathbb{E}\left[F\left(\frac{1}{T}\sum_{t=1}^T{\bm{w}(t)}\right)\right]-\min_{\bm{w}\in\mathbb{R}^P}{F(\bm{w})}\leq\epsilon$, the number of iterations should satisfy
\begin{align}
T\geq A^2\left(\sqrt{\frac{D}{2\epsilon^2}+\frac{\beta}{\epsilon}}+\sqrt{\frac{D}{2\epsilon^2}}\right)^2.
\label{equ:number_of_iterations}
\end{align}
Moreover, the minimum achievable communication cost at iteration $t$ is given by
\begin{align}
P\left[\frac{K}{2}\log\left(1+\frac{\sigma_{N}^2(t)}{KD-\sigma_{N}^2(t)}\right)+\frac{1}{2}\log(1+\frac{\sigma_{X}^2(t)}{D})\right].
\label{equ:lower_bound_communication}
\end{align}
\end{corollary}
\begin{proof}
To guarantee the convergence bound $\mathbb{E}\left[F\left(\frac{1}{T}\sum_{t=1}^T{\bm{w}(t)}\right)\right]-\min_{\bm{w}\in\mathbb{R}^P}{F(\bm{w})}\leq\epsilon$, from Theorem \ref{theorem:convergence_rate}, we should have
\begin{align}
A\sqrt{\frac{2D}{T}}+\frac{\beta A^2}{T}\leq\epsilon\Leftrightarrow T\geq A^2\left(\sqrt{\frac{D}{2\epsilon^2}+\frac{\beta}{\epsilon}}+\sqrt{\frac{D}{2\epsilon^2}}\right)^2.
\end{align}
We can obtain (\ref{equ:lower_bound_communication}) directly by Corollary \ref{corollary:sum_rate_distortion_function}.
\end{proof}
\begin{remark}
The minimum achievable communication cost at iteration $t$ can be potentially achieved by combining vector quantization and Slepian-Wolf coding \cite{4494710}.
For a distortion $D$, the communication cost per iteration can be reduced by distributed source coding without loss of convergence rate.
This is because distributed source coding, such as Slepian-Wolf coding, is lossless, and using it for a quantization scheme does not increase the distortion of the estimator \cite{4494710}.
\end{remark}
\begin{remark}
We can obtain the total communication cost by combining the number of iterations (\ref{equ:number_of_iterations}) and the communication cost at each iteration $t$ (\ref{equ:lower_bound_communication}).
Moreover, given the gradient statistics $\sigma_X^2(t)$ and $\sigma_N^2(t)$ at each iteration $t$, Corollary \ref{corollary:convex_communication_efficiency} can guide the selection of the distortion $D$ in terms of the total communication cost.
In the case of static gradient statistics, it can be found that the monotonicity of the total communication bits depends on $\frac{\beta\epsilon}{\sigma_X^2}$.
When $\frac{\beta\epsilon}{\sigma_X^2}>0.5$, the total communication bits monotonically decrease with the distortion $D$.
In this case, we should choose a low quantization level for quantization schemes to increase the distortion $D$ and hence achieve a low total communication bits.
\end{remark}

In practice, however, the gradient statistics $\sigma_X^2(t)$ and $\sigma_N^2(t)$ shown in Fig.~\ref{fig:statistic_over_iteration} are time-varying and the loss can be a non-convex function.
Hence, we show the effect of distortion on the total communication cost by numerical results.
Note that the sum rate mainly depends on the ratio of distortion $D$ to $\sigma_N^2/K$, and we have $D\geq\sigma_N^2/K$.
Hence, we define a new non-negative parameter $\rho\triangleq\frac{DK}{\sigma_N^2}-1$.
When $\rho>0$, it is the ratio of the variance of quantization noise and the variance of gradient noise, and when $\rho=0$, it indicates error-free transmission of gradients.
\begin{figure}[t]
\begin{centering}
\vspace{-0.2cm}
\includegraphics[scale=.50]{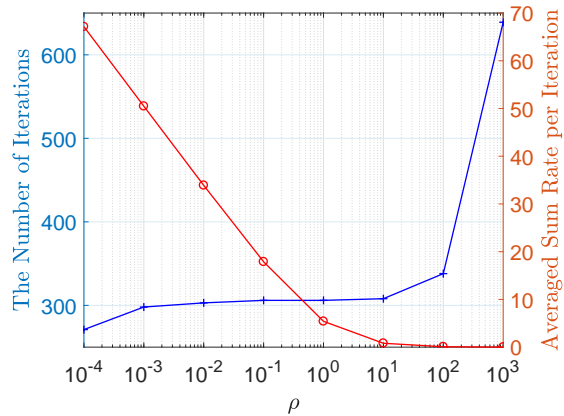}
\vspace{-0.1cm}
 \caption{\small{The number of iterations and averaged sum rate per iteration over $\rho$, where the MLP model achieves loss $1.70$ (accuracy 90\%) on MNIST dataset and the number of worker nodes is 10.}}
 \label{fig:iteration_rate_over_distortion}
\end{centering}
\vspace{-0.3cm}
\end{figure}

Fig.~\ref{fig:iteration_rate_over_distortion} shows the number of iterations and averaged sum rate per iteration over $\rho$, where the MLP model achieves a loss of $1.70$ and accuracy of 90\% on the MNIST dataset.
At each iteration, we add independent quantization noise to each local gradient such that the distortion satisfies $D(t)=\frac{\sigma_N^2(t)}{K}(1+\rho)$.
For each $\rho$, the model is trained until the loss is less than $1.70$ (model converges in this case), and the averaged sum rate per iteration is calculated.
It is observed that the required number of iterations grows slowly over $\rho$ when $\rho<100$, which indicates that moderate variance of quantization noise reduces the convergence rate very slightly.
However, the model convergence rate drops sharply when $\rho>10^2$ and even fails to converge when $\rho>10^3$.
It is also observed that the average sum rate decreases over $\rho$.
Hence, the optimal value of $\rho$ should be selected between $10^1$ and $10^2$ so that both the sum rate and the number of iterations are low.
In practice, the optimal $\rho$ corresponds to the quantization schemes with large distortion such as 1-bit SGD \cite{seide20141} and signSGD \cite{pmlr-v80-bernstein18a}.

\section{Conclusion}
This paper studied the fundamental limit of communication cost of model aggregation in distributed learning from a rate-distortion perspective.
We formulated the model aggregation based on gradient distribution of distributed learning and justified our assumptions by experiments on real-world datasets.
We derived the optimal rate region results and the sum-rate-distortion function for the model aggregation problem.
Based on the derived sum-rate-distortion function and the gradient statistics of real-world datasets, we then numerically analyzed the optimal communication cost at each iteration and the total communication cost to achieve a target model convergence bound.
Numerical results demonstrated that the communication gain from exploiting the correlation between worker nodes is significant and a high variance gradient estimator can achieve a low total communication cost.

Among open problems, since this paper focuses on the theoretical minimum communication cost for model aggregation, future works can study practical distributed source coding to achieve the theoretical gain of model aggregation.
Moreover, the theoretical minimum communication cost of the biased gradient estimator is also an important problem that remains to solve.

\begin{appendices}
\section{Achievability of Theorem \ref{theorem:rate_region}}\label{proof:one_layer_achievability}
For the achievability proof, we adopt the Berger-Tung scheme \cite{10016434852} but design an unbiased estimator at receiver.
\begin{lemma}[Berger-Tung inner bound]
If we can find auxiliary random variables $U_1,U_2,...,U_K$ such that
\begin{align}
U_1,U_2,...,U_K|Y_1,Y_2,...,Y_K\sim p(u_1|y_1)p(u_2|y_2)...p(u_K|y_K)
\end{align}
and decoding function $\hat{X}(U_1,U_2,...,U_K)$ such that
\begin{align}
\mathbb{E}(X-\hat{X}(U^K))^2\leq D,
\end{align}
then the following rate region is achievable
\begin{align}
\sum_{k\in\mathcal{A}}{R_k}\geq I(\bm{Y}_\mathcal{A};\bm{U}_\mathcal{A}|\bm{U}_{\mathcal{A}^c}),\forall\mathcal{A}\subseteq\{1,2,...,K\}.
\end{align}
\end{lemma}

Now, let us apply the Berger-Tung scheme into our problem.
For each worker node $k$, we define the auxiliary random variable $U_k=Y_k+V_k$, where each $V_k\sim\mathcal{N}(0,\sigma_{V_{k}}^2)$ is independently distributed and independent of $Y_k$ and $X$.
Parameters $\{\sigma_{V_{k}}^2\}$ are determined in terms of the given distortion $D$.
After recovering $U^{nK}$, the decoder reconstructs $\hat{X}^n$ by applying the following weighted averaging function component-wise
\begin{align}
\hat{X}=\psi(U^K)\triangleq \sum_{k=1}^K{\alpha_{k}U_{k}},
\end{align}
where $\alpha_{k}=\frac{(\sigma_{N_{k}}^2+\sigma_{V_{k}}^2)^{-1}}{\sum_{k=1}^K{(\sigma_{N_{k}}^2+\sigma_{V_{k}}^2)^{-1}}}$.
Note that $\hat{X}$ is an unbiased estimator of $X$ and we set the variance of $\hat{X}$ equal to the given distortion $D$, which is given by
\begin{align}
D=\left(\sum_{k=1}^K\frac{1}{\sigma_{N_{k}}^2+\sigma_{V_{k}}^2}\right)^{-1}.
\label{proof_direct:D}
\end{align}
Let us define
\begin{align}
r_{k}\triangleq I(Y_{k};U_{k}|X)=\frac{1}{2}\log{(1+\frac{\sigma_{N_{k}}^2}{\sigma_{V_{k}}^2})},k=1,2,...,K.
\end{align}
Parameter $r_{k}$ can be interpreted as the rate of worker node $k$ for quantizing its gradient noise.
We will use $r_{k}$'s as the parameters instead of $\sigma_{V_{k}}^2$ in the following.
Note that for any choice of $(r_{k}\geq 0,k=1,2,...,K)$, we can find a corresponding $(\sigma_{V_{k}}^2\geq 0,k=1,2,...,K)$ and therefore, a set of auxiliary random variables.
Then we can rewrite (\ref{proof_direct:D}) in terms of $r_{k}$'s as
\begin{align}
\label{proof_direct:distortion}
\frac{1}{D}=\sum_{k=1}^K{\frac{1-\exp(-2r_{k})}{\sigma_{N_{k}}^2}},
\end{align}
as desired.

From the Berger-Tung inner bound, $(R_1,...,R_K)$ are achievable if for all non-empty $\mathcal{A}\subseteq\{1,2,...,K\}$ satisfying
\begin{align}
\sum_{k\in\mathcal{A}}{R_k}&\geq I(\bm{Y}_\mathcal{A};\bm{U}_\mathcal{A}|\bm{U}_{\mathcal{A}^c})\\
&=I(\bm{Y}_\mathcal{A},X;\bm{U}_\mathcal{A}|\bm{U}_{\mathcal{A}^c})\\
&=I(X;\bm{U}_\mathcal{A}|\bm{U}_{\mathcal{A}^c})+I(\bm{Y}_\mathcal{A};\bm{U}_\mathcal{A}|X)\\
&=I(X;\bm{U}_\mathcal{A}|\bm{U}_{\mathcal{A}^c})+\sum_{k\in\mathcal{A}}{r_{k}}\\
&=h(X|\bm{U}_{\mathcal{A}^c})-h(X|\bm{U})+\sum_{k\in\mathcal{A}}{r_{k}}\\
&=-\frac{1}{2}\log{\left(\frac{1}{\sigma_{X}^2}+\sum_{k\in\mathcal{A}^c}{\frac{1}{\sigma_{N_{k}}^2+\sigma_{V_{k}}^2}}\right)}+\frac{1}{2}\log{(\frac{1}{\sigma_{X}^2}+\sum_{k\in\mathcal{K}}\frac{1}{\sigma_{N_{k}}^2+\sigma_{V_{k}}^2})}+\sum_{k\in\mathcal{A}}{r_{k}}\\
&=-\frac{1}{2}\log{\left(\frac{1}{\sigma_{X}^2}+\sum_{k\in\mathcal{A}^c}{\frac{1-\exp(-2r_{k})}{\sigma_{N_{k}}^2}}\right)}+\frac{1}{2}\log{(\frac{1}{\sigma_{X}^2}+\frac{1}{D})}+\sum_{k\in\mathcal{A}}{r_{k}}.\label{proof_direct:rate_region}
\end{align}
Equations (\ref{proof_direct:distortion}) and (\ref{proof_direct:rate_region}) together conclude the achievable proof.

\section{Converse of Theorem \ref{theorem:rate_region}}\label{proof:one_layer_converse}
Our converse proof is inspired by Oohama's converse proof in \cite{669162} but is not straightforward and cannot be derived directly from traditional (biased) CEO.
The existence of unbiased estimator must be guaranteed in the unbiased CEO problem, which is nontrivial since unbiased estimator might not always exist for arbitrary descriptions.

Suppose we achieve an unbiased estimator $\hat{X}^n$ with distortion $D=D(X^n,\hat{X}^n)$.
Let $\bm{C}_\mathcal{K}=(C_1,C_2,...,C_K)$ denote all the messages produced by all the worker nodes after observing $n$-length sequences.
Let us define
\begin{align}
r_{k}\triangleq\frac{1}{n}I(Y_{k}^n;C_k|X^n).
\end{align}
For any $\mathcal{A}\subseteq\{1,2,...,K\}$, we have
\begin{align}
\sum_{k\in\mathcal{A}}{R_k}&\geq\sum_{k\in\mathcal{A}}{\frac{1}{n}H(C_k)}\\
&\geq\frac{1}{n}H(\bm{C}_\mathcal{A})\\
&\geq\frac{1}{n}H(\bm{C}_\mathcal{A}|\bm{C}_{\mathcal{A}^c})\label{proof_converse:rate_region_line2}\\
&\geq\frac{1}{n}I(\bm{Y}_\mathcal{A}^n;\bm{C}_\mathcal{A}|\bm{C}_{\mathcal{A}^c})\\
&=\frac{1}{n}I(X^n,\bm{Y}_\mathcal{A}^n;\bm{C}_\mathcal{A}|\bm{C}_{\mathcal{A}^c})\label{proof_converse:rate_region_line4}\\
&=\frac{1}{n}I(X^n;\bm{C}_\mathcal{A}|\bm{C}_{\mathcal{A}^c})+\frac{1}{n}I(\bm{Y}_\mathcal{A}^n;\bm{C}_\mathcal{A}|\bm{C}_{\mathcal{A}^c},X^n)\label{proof_converse:rate_region_line5}\\
&=\frac{1}{n}I(X^n;\bm{C}_\mathcal{A}|\bm{C}_{\mathcal{A}^c})+\frac{1}{n}I(X^n;\bm{C}_{\mathcal{A}^c})-\frac{1}{n}I(X^n;\bm{C}_{\mathcal{A}^c})+\frac{1}{n}I(\bm{Y}_\mathcal{A}^n;\bm{C}_\mathcal{A}|X^n)\label{proof_converse:rate_region_line6}\\
&=\frac{1}{n}I(X^n;\bm{C}_\mathcal{K})-\frac{1}{n}I(X^n;\bm{C}_{\mathcal{A}^c})+\sum_{k\in\mathcal{A}}\frac{1}{n}I(Y_k^n;C_k|X^n)\label{proof_converse:rate_region_line7}\\
&=\frac{1}{n}I(X^n;\bm{C}_\mathcal{K})-\frac{1}{n}I(X^n;\bm{C}_{\mathcal{A}^c})+\sum_{k\in\mathcal{A}}r_{k}\label{proof_converse:converse_origin},
\end{align}
where (\ref{proof_converse:rate_region_line2}) follows from the fact that conditioning reduces entropy, (\ref{proof_converse:rate_region_line4}) follows from the fact that $X^n-\bm{Y}_\mathcal{A}^n-\bm{C}_\mathcal{A}$ forms a Markov chain, (\ref{proof_converse:rate_region_line5}) follows from the chain rule for mutual information, (\ref{proof_converse:rate_region_line6}) follows from the fact that $(\bm{Y}_\mathcal{A}^n,\bm{C}_\mathcal{A})-X^n-\bm{C}_{\mathcal{A}^c}$ forms a Markov chain, (\ref{proof_converse:rate_region_line7}) follows from the chain rule for mutual information and the fact that $Y_k^n-X^n-\bm{C}_{\mathcal{K}\setminus\{k\}}$ forms a Markov chain.
In the following, we bound the first two terms in (\ref{proof_converse:converse_origin}), respectively.
In order to bound the first term in (\ref{proof_converse:converse_origin}), we first introduce the following lemma.
\begin{lemma}
\label{lemma:first_term}
Let $\hat{X}^n$ be any unbiased estimator of $X^n$ with distortion $D(X^n,\hat{X}^n)$.
We have
\begin{align}
\frac{1}{n}I(X^n;\hat{X}^n)\geq\frac{1}{2}\log(1+\frac{\sigma_X^2}{D(X^n,\hat{X}^n)}).
\end{align}
\end{lemma}
\begin{proof}
We have
\begin{align}
\frac{1}{n}I(X^n;\hat{X}^n)&=\frac{1}{n}h(X^n)-\frac{1}{n}h(X^n|\hat{X}^n)\\
&=\frac{1}{n}h(X^n)-\frac{1}{n}\sum_{t=1}^n{h(X(t)|X^{t-1}\hat{X}^n)}\\
&\geq\frac{1}{n}h(X^n)-\frac{1}{n}\sum_{t=1}^n{h(X(t)|\hat{X}(t))}\\
&\geq\frac{1}{2}\log2\pi e\sigma_X^2-\frac{1}{n}\sum_{t=1}^n{\frac{1}{2}\log2\pi e\sigma_{X(t)|\hat{X}(t)}^2}\label{proof_lemma:middle}.
\end{align}
The conditional variance of variable $X$ given the unbiased estimator $\hat{X}$ is given by
\begin{align}
\sigma_{X|\hat{X}}^2&=\mathbb{E}\left[X-\mathbb{E}(X|\hat{X})\right]^2\\
&\leq\min_a{\mathbb{E}(X-a\hat{X})^2}\\
&=\min_a{\mathbb{E}\left[\mathbb{E}\left((X-a\hat{X})^2|X\right)\right]}\\
&=\min_a{\mathbb{E}\left[\mathbb{E}^2(X-a\hat{X}|X)+\mbox{Var}(X-a\hat{X}|X)\right]}\\
&=\min_a{\mathbb{E}\left\{(X-aX)^2+\mathbb{E}\left[(X-a\hat{X})-\mathbb{E}(X-a\hat{X}|X)\right]^2\right\}}\label{proof_lemma:conditional_expectation1}\\
&=\min_a{\left[(1-a)^2\sigma_X^2+a^2\mathbb{E}(X-\hat{X})^2\right]}\label{proof_lemma:conditional_expectation2}\\
&=(\frac{1}{\sigma_X^2}+\frac{1}{\mathbb{E}(X-\hat{X})^2})^{-1}\label{proof_lemma:variance}.
\end{align}
where (\ref{proof_lemma:conditional_expectation1}) and (\ref{proof_lemma:conditional_expectation2}) follow from the fact that $\mathbb{E}(\hat{X}|X)=X$.
Substituting (\ref{proof_lemma:variance}) to (\ref{proof_lemma:middle}), we have
\begin{align}
\frac{1}{n}I(X^n;\hat{X}^n)&\geq\frac{1}{2}\log2\pi e\sigma_X^2+\frac{1}{n}\sum_{t=1}^n{\frac{1}{2}\log\frac{1}{2\pi e}(\frac{1}{\sigma_X^2}+\frac{1}{\mathbb{E}(X(t)-\hat{X}(t))^2})}\\
&\geq\frac{1}{2}\log2\pi e\sigma_X^2+\frac{1}{2}\log\frac{1}{2\pi e}(\frac{1}{\sigma_X^2}+\frac{1}{D(X^n,\hat{X}^n)})\label{proof_lemma:jasen_inequality}\\
&=\frac{1}{2}\log(1+\frac{\sigma_X^2}{D(X^n,\hat{X}^n)}),
\end{align}
where (\ref{proof_lemma:jasen_inequality}) follows from Jensen's inequality.
The proof of Lemma \ref{lemma:first_term} is completed.
\end{proof}

Compared to the lower bound $\frac{1}{2}\log(\frac{\sigma_X^2}{D})$ obtained in \cite{1365154}, Eq.(8), Lemma \ref{lemma:first_term} uses the unbiased estimator to tighten it to $\frac{1}{2}\log(1+\frac{\sigma_X^2}{D})$.
Based on Lemma 3, we have a simple lower-bound for the first term in (\ref{proof_converse:converse_origin}).
\begin{align}
\frac{1}{n}I(X^n;\bm{C}_\mathcal{K})&\geq\frac{1}{n}I(X^n;\hat{X}^n)\label{proof_converse:first_term_line1}\\
&\geq\frac{1}{2}\log(1+\frac{\sigma_{X}^2}{D(X^n,\hat{X}^n)})\label{proof_converse:first_term_line2}\\
&=\frac{1}{2}\log(1+\frac{\sigma_{X}^2}{D})\label{proof_converse:first_term},
\end{align}
where (\ref{proof_converse:first_term_line1}) follows from the data-processing inequality and (\ref{proof_converse:first_term_line2}) follows from Lemma \ref{lemma:first_term}.

We introduce the following lemma to bound the second term in (\ref{proof_converse:converse_origin}), whose proof is given in \cite{1365154}, Lemma 3.1.
\begin{lemma}
\label{lemma:second_term}
Let $\mathcal{A}\subseteq\{1,2,...,K\}$.
Then we have
\begin{align}
\frac{1}{\sigma_{X}^2}\exp\left(\frac{2}{n}I(X^n;\bm{C}_\mathcal{A})\right)\leq\frac{1}{\sigma_{X}^2}+\sum_{k\in\mathcal{A}}\frac{1-\exp(-2r_{k})}{\sigma_{N_{k}}^2}.
\label{proof_converse:second_term}
\end{align}
\end{lemma}
Combining (\ref{proof_converse:first_term}), (\ref{proof_converse:second_term}) and (\ref{proof_converse:converse_origin}), for all $\mathcal{A}\subseteq\{1,2,...,K\}$, we have
\begin{align}
\sum_{k\in\mathcal{A}}{R_k}&\geq\frac{1}{2}\log(1+\frac{\sigma_{X}^2}{D})-\frac{1}{2}\log\left(1+\sigma_{X}^2\sum_{k\in{\mathcal{A}^c}}\frac{1-\exp(-2r_{k})}{\sigma_{N_{k}}^2}\right)+\sum_{k\in\mathcal{A}}r_{k}\\
&=-\frac{1}{2}\log\left(\frac{1}{\sigma_{X}^2}+\sum_{k\in\mathcal{A}^c}{\frac{1-\exp(-2r_{k})}{\sigma_{N_{k}}^2}}\right)+\frac{1}{2}\log{(\frac{1}{\sigma_{X}^2}+\frac{1}{D})}+\sum_{k\in\mathcal{A}}{r_{k}}.
\label{proof_converse:converse_replaced}
\end{align}
Substituting (\ref{proof_converse:first_term}) to (\ref{proof_converse:second_term}) with $\mathcal{A}=\mathcal{K}$, we have the condition
\begin{align}
\sum_{k=1}^K{\frac{1-\exp(-2r_{k})}{\sigma_{N_{k}}^2}}\geq\frac{1}{D}.
\label{proof_converse:r_constrain}
\end{align}
Along with the non-negativity constraints on $r_{k}$'s, (\ref{proof_converse:converse_replaced}) and (\ref{proof_converse:r_constrain}) define an outer bound for $\mathcal{R}_\star(D)$.
It is easy to show that replacing the inequality in (\ref{proof_converse:r_constrain}) with an equality will not change this outer bound.
Thus we have $\mathcal{R}_\star(D)\subseteq\mathcal{R}_N(D)$.

\section{Proof of Corollary \ref{corollary:rate_region}}\label{proof:corollary:rate_region}
We first prove the achievability.
Given a distortion $D$, for all $(D_1,D_2,\ldots,D_L)\in\mathbb{R}_+^L$ satisfying $\sum_{l=1}^L{D_l}=D$, we apply the achievable scheme in Theorem \ref{theorem:rate_region} to encode and reconstruct each layer $l$ separately such that the layer $l$'s distortion satisfies $D(X_l^n,\hat{X}_l^n)\leq D_l$.
Specifically, for each layer $l\in\{1,2,\ldots,L\}$ and its corresponding distortion $D_l$, for all $(r_{1,l},r_{2,l},\ldots,r_{K,l})\in\mathcal{F}_l(D_l)$, Theorem \ref{theorem:rate_region} states that all the rate $(R_{1,l},\ldots,R_{K,l})\in\mathcal{R}_l(r_{1,l},r_{2,l},\ldots,r_{K,l})$ is achievable, i.e., there exists encoders for layer $l$ with the rate $(R_{1,l},\ldots,R_{K,l})$ such that the layer $l$'s distortion satisfies $D(X_l^n,\hat{X}_l^n)\leq D_l$.
By combining all the layers' encoders, the rate of each worker node $k$, for $k=1,2,\ldots,K$, is $R_k=\sum_{l=1}^L{R_{k,l}}$, and we have
\begin{align}
D(\bm{X}^n,\bm{\hat{X}}^n)=\sum_{l=1}^L{D(X_l^n,\hat{X}_l^n)}\leq\sum_{l=1}^L{D_l}=D.
\end{align}
Hence, we have completed the achievable proof of this corollary.

Now, we prove the converse.
Suppose we achieve an unbiased gradient estimator $\bm{\hat{X}}^n$ with distortion $D=\sum_{l=1}^L{D_l}$, where $D_l=D(X_l^n,\hat{X}_l^n)$.
Let $\bm{C}_\mathcal{K}=(C_1,C_2,...,C_K)$ denote all the messages produced by all the worker nodes after observing $n$-length sequences.
Let us define
\begin{align}
r_{k,l}\triangleq\frac{1}{n}I(Y_{k,l}^n;C_k|X_l^n).
\end{align}
For any $\mathcal{A}\subseteq\{1,2,...,K\}$, we have
\begin{align}
\sum_{k\in\mathcal{A}}{R_k}&\geq\frac{1}{n}I(\bm{X}^n;\bm{C}_\mathcal{K})-\frac{1}{n}I(\bm{X}^n;\bm{C}_{\mathcal{A}^c})+\sum_{k\in\mathcal{A}}\frac{1}{n}I(\bm{Y}_k^n;C_k|\bm{X}^n)\label{proof_converse:rate_region_vector}\\
&=\sum_{l=1}^L{\left(\frac{1}{n}I(X_l^n;\bm{C}_\mathcal{K})-\frac{1}{n}I(X_l^n;\bm{C}_{\mathcal{A}^c})+\sum_{k\in\mathcal{A}}r_{k,l}\right)}\\
&\geq\sum_{l=1}^L\left(-\frac{1}{2}\log\left(\frac{1}{\sigma_{X_l}^2}+\sum_{k\in\mathcal{A}^c}{\frac{1-\exp(-2r_{k,l})}{\sigma_{N_{k,l}}^2}}\right)+\frac{1}{2}\log{(\frac{1}{\sigma_{X_l}^2}+\frac{1}{D_l})}+\sum_{k\in\mathcal{A}}{r_{k,l}}\right),
\label{proof_converse:converse_replaced_vector}
\end{align}
where (\ref{proof_converse:rate_region_vector}) is similar to the derivation of (\ref{proof_converse:rate_region_line7}), and (\ref{proof_converse:converse_replaced_vector}) follows from Lemma \ref{lemma:first_term} and Lemma \ref{lemma:second_term}.
For each $l=1,2,...,L$, similar to the derivation of (\ref{proof_converse:r_constrain}), we have the condition
\begin{align}
\sum_{k=1}^K{\frac{1-\exp(-2r_{k,l})}{\sigma_{N_{k,l}}^2}}\geq\frac{1}{D_l},l=1,2,...,L.
\label{proof_converse:r_constrain_vector}
\end{align}
It is easy to show that replacing the inequality in (\ref{proof_converse:r_constrain_vector}) with an equality will not change this outer bound.
Thus we have completed the proof of the converse part.

\section{Proof of Lemma \ref{lemma:optimal_R_kl}}\label{proof:lemma:optimal_R_kl}
Note that the equation (\ref{equ:optimal_R_kl}) is equivalent to $R_{k,l}=I(Y_{k,l};U_{k,l}|U_{k+1,l},...,U_{K,l})$.
Based on the definition of $\mathcal{R}_{l}(r_{1,l},...,r_{K,l})$, we can easily know that $(R_{1,l},...,R_{K,l})\in\mathcal{R}_{l}(r_{1,l},...,r_{K,l})$.
In the following, we prove that $(R_{1,l},R_{2,l},...,R_{K,l})$ is the optimal choice in $\mathcal{R}_{l}(r_{1,l},r_{2,l},...,r_{K,l})$ for minimizing $\sum_{k\in\mathcal{K}}{\alpha_k R_{k,l}}$.

To prove that $(R_{1,l},R_{2,l},...,R_{K,l})$ is the optimal choice in $\mathcal{R}_{l}(r_{1,l},r_{2,l},...,r_{K,l})$ is equal to prove that for all $(R_{1,l}^{'},R_{2,l}^{'},...,R_{K,l}^{'})\in\mathcal{R}_{l}(r_{1,l},r_{2,l},...,r_{K,l})$ we have $\sum_{k\in\mathcal{K}}{\alpha_k R_{k,l}}\leq\sum_{k\in\mathcal{K}}{\alpha_k R_{k,l}^{'}}$.
For all $(R_{1,l}^{'},R_{2,l}^{'},...,R_{K,l}^{'})\in\mathcal{R}_{l}(r_{1,l},r_{2,l},...,r_{K,l})$, based on the definition of $\mathcal{R}_{l}(r_{1,l},r_{2,l},...,r_{K,l})$, we have
\begin{align}
\sum_{i=1}^k{R_{i,l}^{'}}\geq I(Y_{1,l},...,Y_{k,l};U_{1,l},...,U_{k,l}|U_{k+1,l},...,U_{K,l})=\sum_{i=1}^k{R_{i,l}},k=1,2,...,K.
\label{lemma:sum_large_sum}
\end{align}
Based on (\ref{lemma:sum_large_sum}), we have
\begin{align}
&\sum_{k\in\mathcal{K}}{\alpha_k R_{k,l}^{'}}-\sum_{k\in\mathcal{K}}{\alpha_k R_{k,l}}\\
=&\sum_{k\in\mathcal{K}}{\alpha_k(R_{k,l}^{'}-R_{k,l})}\\
=&\sum_{k\in\mathcal{K}}{\left[(\alpha_k-\alpha_{k+1})+(\alpha_{k+1}-\alpha_{k+2})+...+(\alpha_{K-1}-\alpha_{K})+(\alpha_{K})\right](R_{k,l}^{'}-R_{k,l})}\\
=&\sum_{k=1}^{K-1}{(\alpha_k-\alpha_{k+1})(\sum_{i=1}^k{R_{i,l}}-\sum_{i=1}^k{R_{i,l}^{'}})}+\alpha_{K}(\sum_{i\in\mathcal{K}}{R_{i,l}}-\sum_{i\in\mathcal{K}}{R_{i,l}^{'}})\\
\geq& 0\label{equ:not_better}.
\end{align}
The inequality in (\ref{equ:not_better}) holds as $\alpha_1\geq\alpha_2\geq...\geq\alpha_K$ and $\sum_{i=1}^k{R_{i,l}^{'}}\geq\sum_{i=1}^k{R_{i,l}}$ for all $k=1,2,...,K$. This indicates that $(R_{1,l},R_{2,l},...,R_{K,l})$ is the optimal solution in $\mathcal{R}_{l}(r_{1,l},r_{2,l},...,r_{K,l})$ that minimizes $\sum_{k\in\mathcal{K}}{\alpha_k R_{k,l}}$.
The proof of Lemma \ref{lemma:optimal_R_kl} is completed.

\section{Proof of Corollary \ref{corollary:sum_rate_distortion_function}}\label{proof:corollary:sum_rate_distortion_function}
The proof is mainly based on Theorem \ref{theorem:rate_region}.
By the definition in (\ref{definition:sum_rate_distortion_function}), we have
\begin{align}
R_{sum}(D)&\triangleq\min_{(R_1,R_2,...,R_K)\in\mathcal{R}_\star(D)}{\sum_{k\in\mathcal{K}}{R_k}}\\
&=\min_{r_{k}}\min_{R_{k}}{\sum_{k\in\mathcal{K}}{R_{k}}}\\
&=\min_{r_{k}}\sum_{k\in\mathcal{K}}r_{k}+\frac{1}{2}\log(1+\frac{\sigma_{X}^2}{D})\label{corollary:sum_rate_distortion_function:sum_R}\\
&=\frac{K}{2}\log\left(1+\frac{\sigma_{N}^2}{KD-\sigma_{N}^2}\right)+\frac{1}{2}\log(1+\frac{\sigma_{X}^2}{D}),\label{corollary:sum_rate_distortion_function:sum_r}
\end{align}
where (\ref{corollary:sum_rate_distortion_function:sum_R}) is obtained by setting $\mathcal{A}=\mathcal{K}$ in (\ref{region:set_R}), and (\ref{corollary:sum_rate_distortion_function:sum_r}) is obtained by substituting $\sum_{k\in\mathcal{K}}r_{k}$ in (\ref{region:set_r}) into (\ref{corollary:sum_rate_distortion_function:sum_R}).
The proof of Corollary \ref{corollary:sum_rate_distortion_function} is completed.
\end{appendices}

\bibliographystyle{IEEEtran}
\bibliography{IEEEabrv,rate-distortion}
\end{document}